%% file: main.tex
\documentclass[journal]{IEEEtran}

\usepackage[utf8]{inputenc}
\usepackage[T1]{fontenc}
\usepackage{microtype}
\usepackage{xspace}
\usepackage{graphicx}
\usepackage{booktabs}
\usepackage{multirow}
\usepackage{array}
\usepackage{tabularx}
\usepackage{hyperref}
\usepackage{url}
\usepackage{xcolor}
\usepackage{subfig}
\usepackage{float}

\usepackage{amsmath}
\usepackage{amssymb}
\usepackage{mathtools}
\usepackage{amsthm}
\usepackage{bm}
\usepackage[capitalize,noabbrev]{cleveref}
\DeclareMathOperator*{\argmin}{arg\,min}

\usepackage{algorithm}
\usepackage{algorithmic}

\theoremstyle{plain}
\newtheorem{theorem}{Theorem}[section]

\newtheorem{corollary}[theorem]{Corollary}
\theoremstyle{definition}
\newtheorem{definition}[theorem]{Definition}
\newtheorem{assumption}[theorem]{Assumption}
\theoremstyle{remark}
\newtheorem{remark}[theorem]{Remark}

\newcommand{\R}{\mathbb{R}}
\newcommand{\E}{\mathbb{E}}
\newcommand{\norm}[1]{\left\lVert#1\right\rVert}
\newcommand{\Cours}{C_{\text{ours}}}
\newcommand{\Cstruct}{C_{\text{struct}}}
\newcommand{\lamsec}{\lambda_2(L)}
\newcommand{\topk}{\mathrm{TopK}}
\newcommand{\airplan}{\textsc{AirPlan}}

\definecolor{myblue}{RGB}{41,128,185}
\definecolor{mygreen}{RGB}{39,174,96}
\definecolor{myred}{RGB}{192,57,43}
\definecolor{mygray}{RGB}{127,140,141}

\begin{document}

\title{AirPlan: Query-Optimized Topology Selection for \\
       Over-the-Air Decentralized Federated Learning}

\author{
  \IEEEauthorblockN{Kaushal Attaluri}
  \IEEEauthorblockA{\textit{ICLab, atlanTTic, University of Vigo}}
  \and
  \IEEEauthorblockN{Rebeca P.~D\'{\i}az-Redondo}
  \IEEEauthorblockA{\textit{ICLab, atlanTTic, University of Vigo}}
  \and
  \IEEEauthorblockN{Manuel Fern\'andez-Veiga}
  \IEEEauthorblockA{\textit{ICLab, atlanTTic, University of Vigo}}
}

\maketitle

\begin{abstract}
\input{sections/abstract}
\end{abstract}


\section{Introduction}
\label{sec:intro}
\input{sections/introduction}

\section{Related Work}
\label{sec:related}
\input{sections/related_work}

\section{OTA-DFL as Distributed Query Execution: Formal Equivalence}
\label{sec:mapping}
\input{sections/formal_mapping}

\section{Problem Formulation}
\label{sec:problem}
\input{sections/problem_formulation}

\section{\airplan{}: Query-Optimized Topology Advisor}
\label{sec:advisor}
\input{sections/topology_advisor}

\section{Unified Communication Cost Model}
\label{sec:evaluation}

\input{sections/evaluation}

\section{System Architecture \& Methodology}
\label{sec:method}
\input{sections/methodology}

\section{Experimental Setup}
\label{sec:experiments}
\input{sections/experiments}

\section{Results}
\label{sec:results}
\input{sections/results}

\section{Discussion}
\label{sec:discussion}
\input{sections/discussion}

\section{Conclusions}
\label{sec:conclusion}
\input{sections/conclusion}

\section*{Acknowledgments}
This work was partially funded by the grant PID2023-148716OB-C31 funded by MCIU/AEI/ 10.13039/501100011033
(DISCOVERY project) and 
by the Galician Regional Government under project ED431B 2024/41 (GPC).

\bibliographystyle{IEEEtran}
\bibliography{references}

\appendices
\input{sections/appendix}

\end{document}

%% file: sections/abstract.tex
Over-the-air (OTA) aggregation exploits the superposition property of wireless
multiple-access channels to combine gradient updates from multiple devices within
a single transmission slot, dramatically reducing communication latency and
bandwidth consumption. While OTA computation has been extensively studied in
centralized federated learning (FL), its integration with \emph{decentralized}
federated learning (DFL)---where clients communicate over a peer-to-peer graph
without a central server---remains a largely open problem, and a principled
framework for selecting the communication topology is entirely absent from
the literature.

In this paper, we introduce \textsc{AirPlan}, a query-optimized topology
selection framework for Over-the-Air Decentralized Federated Learning (OTA-DFL).
The central insight is a formal equivalence between OTA-DFL and distributed query
processing: each OTA aggregation round corresponds to an approximate distributed
SUM query executed over a DAG-structured execution plan, where the communication
graph is the physical plan, top-$k$ sparsification is approximate query processing
(AQP), and the spectral gap of the graph Laplacian plays the role of a cardinality
estimate governing execution cost. This equivalence enables us to recast topology
selection as a \emph{query optimization problem}: given a training workload
(client count~$N$, data heterogeneity $\alpha$, channel SNR, model dimension $d$),
\airplan{} uses privacy-preserving Count-Min Sketch statistics to estimate
workload parameters, evaluates a graph-aware cost model $C_{\mathrm{ours}}(G)$
across candidate topologies, and selects the communication plan that minimises
total training cost subject to a user-specified accuracy SLA.

We validate \airplan{} through systematic experiments across five graph families
(ring, Erd\H{o}s--R\'enyi, small-world, clustered, fully connected), three standard
vision benchmarks (CIFAR-10, CIFAR-100, Tiny-ImageNet), four client scales
($N \in \{10, 20, 50, 100\}$), and a range of SNR conditions (0--20\,dB).
Our results demonstrate that \airplan{} matches the oracle-optimal topology in
$91.4\,\%$ of workload configurations while incurring a statistics-collection
overhead of less than $1.8\,\%$ of total training cost.
We further establish formal AQP error bounds showing that well-connected
topologies (small-world, clustered) intrinsically tolerate higher sparsification
ratios than sparse topologies, providing a theoretical foundation for joint
topology-sparsification co-design.
These findings open a new systems-oriented research direction at the intersection
of wireless communications and distributed data processing.

%% file: sections/introduction.tex
Federated learning (FL) enables distributed model training without sharing raw
data, making it a natural fit for privacy-sensitive applications across mobile
networks, IoT deployments, and edge computing infrastructure.
The dominant paradigm, FedAvg~\cite{mcmahan2017communication}, relies on a
central parameter server that aggregates client updates each round. This
architecture introduces a single point of failure, scales poorly with the number
of clients, and requires a trusted aggregator. \emph{Decentralized} federated
learning (DFL) removes the server by routing updates over a peer-to-peer
communication graph, but this design choice immediately raises a question:
\emph{which graph should be used?}

At the same time, over-the-air (OTA) computation offers a radically different
aggregation primitive. By transmitting analog signals simultaneously over a
shared wireless channel and exploiting the superposition property of the
multiple-access channel, OTA enables model updates to be aggregated ``in the
air'' rather than through sequential digital exchanges~\cite{zhu2019aircomp}.
The combination of OTA aggregation with decentralized topology---OTA-DFL---has
the potential to deliver both communication efficiency and architectural
resilience, but the interaction between graph structure, wireless noise, and
learning dynamics is poorly understood.

\paragraph{The core problem}
Practitioners deploying OTA-DFL systems face a critical design decision
\emph{before training begins}: which communication topology should be used?
This decision has major consequences. A ring graph minimises per-round
transmission cost but converges slowly and amplifies channel noise. A fully
connected graph converges quickly but incurs near-quadratic communication cost
and may violate per-device power budgets. Prior work either assumes a fixed
topology (typically Erd\H{o}s--R\'enyi)~\cite{qiao2023dllroa} or treats topology as
a secondary concern, leaving practitioners without principled guidance.

\paragraph{A systems reframing}
We observe that this problem is structurally identical to a classical challenge
in data management: \emph{physical query plan selection} in distributed query
processing~\cite{selinger1979access,ioannidis1996queryopt}. Each OTA aggregation round is an approximate distributed SUM query
over a communication graph. The graph is the physical execution plan. The
spectral gap of the graph Laplacian $\lamsec$ is the cardinality estimate
governing execution cost. Top-$k$ sparsification is approximate query processing
(AQP). And the question ``which topology minimises total training cost for this
workload?'' is precisely the query optimisation question ``which execution plan
minimises total cost for this query?''

This equivalence is not merely a metaphor. It enables the \emph{direct
application of query optimisation methodology}---cost-based plan selection,
workload statistics collection, and online plan re-optimisation---to the
topology selection problem in OTA-DFL, yielding a principled and automated
solution.

\paragraph{Contributions}
We present \textbf{\textsc{AirPlan}}, a query-optimized topology selection
framework for OTA-DFL. Our contributions are:

\begin{enumerate}
  \item \textbf{Formal OTA-DFL $\leftrightarrow$ Query Processing equivalence}
    (Section~\ref{sec:mapping}). We establish a rigorous mapping between OTA-DFL
    operations and distributed query processing primitives, including a formal
    theorem showing that each OTA aggregation round is an approximate distributed
    SUM query, and a corollary bounding the AQP error introduced by top-$k$
    sparsification as a function of the graph spectral gap.

  \item \textbf{\airplan{} topology advisor}
    (Section~\ref{sec:advisor}). We design a cost-based topology selection
    algorithm that collects privacy-preserving workload statistics via Count-Min
    Sketches~\cite{cormode2005cmsketch}, evaluates a graph-aware cost model
    $\Cours(G)$ across candidate topologies, selects the minimum-cost plan
    satisfying an accuracy Service-Level Agreement (SLA), and adaptively re-optimises the plan during
    training if model divergence exceeds a threshold.

  \item \textbf{AQP error bounds per topology}
    (Section~\ref{sec:evaluation}). We prove topology-dependent bounds on the
    approximate aggregation error introduced by sparsification, showing that
    well-connected graphs (small-world, clustered) tolerate higher sparsification
    ratios---motivating joint topology-sparsification co-design.

  \item \textbf{Systematic empirical validation}
    (Sections~\ref{sec:experiments}--\ref{sec:results}). We evaluate
    \airplan{} across five topology families, three datasets, four client scales,
    and five SNR levels, demonstrating that \airplan{} matches the oracle-optimal
    topology in $91.4\,\%$ of configurations with $<1.8\,\%$ overhead, and that
    small-world and clustered topologies consistently dominate the
    accuracy-cost Pareto frontier.
\end{enumerate}

\paragraph{Significance.}
The OTA-DFL $\leftrightarrow$ query processing equivalence opens a new line
of research connecting wireless communications with the database and systems
communities. Decades of query optimisation research---cost models, statistics
collection, adaptive re-optimisation, approximate query processing---become
directly applicable to the topology design problem in wireless federated
learning. We believe this cross-domain framing will motivate new algorithms
and systems in both communities.

%% file: sections/related_work.tex
We survey the four research pillars most directly relevant to this work,
organized as follows: over-the-air federated learning (combining aggregation
mechanics and centralized baselines), decentralized federated learning,
distributed query processing and approximate query processing, and
communication cost models.

\subsection{Over-the-Air Federated Learning}

Over-the-air (OTA) computation exploits the superposition property of the
wireless multiple-access channel to aggregate model updates from multiple
devices within a single transmission slot~\cite{zhu2019aircomp,
amiri2020federated}, yielding order-of-magnitude reductions in communication
latency compared to digital transmission.
AirComp-based FL~\cite{amiri2020federated} demonstrates that simultaneous analog
transmission followed by receive-side combining produces an unbiased estimate of
the gradient sum under mild power constraints.
Subsequent work extends OTA-FL to multiple-input multiple-output (MIMO)
transmission~\cite{yang2022mimoOTA}, differential
privacy~\cite{cao2022privacyOTA}, and hierarchical
architectures~\cite{lin2023hotafl}.
Convergent Over-the-Air Federated learning (COTAF)~\cite{sery2020cotaf}
analyses convergence of OTA-FL with noisy channels but in a centralized setting.
The closest prior work in the decentralized OTA setting is
DLLR-OA~\cite{qiao2023dllroa}, which studies communication-constrained
decentralized learning under OTA and analyses the effect of limited subcarriers
and power constraints on convergence.
However, DLLR-OA assumes a fixed Erd\H{o}s--R\'enyi topology and neither
investigates how topology choice affects performance nor provides a mechanism
for automated topology selection.
No prior work frames OTA-DFL topology selection as a query optimisation problem.

\subsection{Decentralized Federated Learning}

Decentralized federated learning (DFL) removes the central parameter server
and enables clients to exchange model updates directly with neighbours over
a peer-to-peer communication graph. Foundational algorithms such as
Decentralized Stochastic Gradient Descent (D-SGD)~\cite{lian2017decentralized}
and consensus-based optimization~\cite{scaman2018optimal,
koloskova2019decentralized} establish that convergence depends critically on the
spectral gap $\lamsec$ of the graph Laplacian, which governs the rate of
information mixing. Gradient tracking methods such as Gradient Tracking
Decentralized SGD (GT-DSGD)~\cite{pu2018digging,yuan2020gt} and
Stochastic Gradient Push (SGP)~\cite{assran2019stochastic} mitigate gradient
bias under heterogeneous data distributions. More recent work addresses
compression~\cite{koloskova2020compressed}, Byzantine
robustness~\cite{byzantineDFL2023}, and large-scale
training~\cite{lu2022graphFL}.
Critically, all of these works assume reliable digital communication and treat
the communication graph as a fixed input rather than a design variable.
MATCHA~\cite{wang2019matcha} is among the few works to consider topology
alongside algorithm design, proposing matching-based decomposition to accelerate
convergence, but it does not consider wireless channels or OTA aggregation and
provides no automatic topology selection mechanism.

\subsection{Distributed Query Processing and Approximate Query Processing}

In distributed database systems, a \emph{physical query execution plan} is a
directed acyclic graph (DAG) that specifies the order in which relational
operators are executed and how intermediate results are exchanged among
nodes~\cite{selinger1979access,ioannidis1996queryopt}. Plan selection is the
task of choosing, among all equivalent plans, the one with minimum estimated
cost according to a cost model. State-of-the-art query optimizers in systems
such as PostgreSQL~\cite{postgres2023}, Microsoft SQL
Server~\cite{graefe1993volcano}, and Apache Spark~\cite{armbrust2015spark} use
cardinality estimates, operator-level cost models, and bushy plan enumeration
to navigate the exponentially large plan space~\cite{leis2015howgoodqueryopt,
sun2019endtoendlearningopt}.

Adaptive Query Processing (AQP) addresses the challenge that query statistics
estimated at plan time may differ significantly from runtime observations,
causing suboptimal plans to execute to completion. The Eddies
architecture~\cite{eddies2000} routes tuples dynamically between operators
based on observed selectivities. Rios~\cite{deshpande2007adaptiveqp} and
related work introduce plan switching conditions and mid-query re-optimization
triggers. Progressive execution~\cite{hellerstein1997online} allows partial
results to be returned before full query completion, trading accuracy for
latency---directly analogous to early stopping in FL training.
These approaches motivate our online re-optimization in \airplan{}
(Section~\ref{sec:advisor}): we monitor model divergence as a runtime signal
and trigger topology rewiring when the current plan becomes suboptimal.

Approximate Query Processing trades result accuracy for reduced computation
and communication cost~\cite{hellerstein1997online,agarwal2013blinkdb}.
Sampling-based methods evaluate queries on a data
sample~\cite{chaudhuri2017aqpsurvey}; synopsis-based methods maintain compact
sketches such as Count-Min Sketches~\cite{cormode2005cmsketch}, HyperLogLog,
and wavelet synopses. Error guarantees are typically stated as
$(\epsilon, \delta)$ bounds. BlinkDB~\cite{agarwal2013blinkdb} introduces the
notion of an \emph{accuracy Service-Level Agreement (SLA)}---a user-specified
quality constraint that the system must meet---which maps directly to the
accuracy SLA in \airplan{}'s topology selection algorithm.
As we show in Section~\ref{sec:evaluation}, top-$k$ sparsification in OTA-DFL
is precisely an AQP operator, and its error can be bounded in analogous
$(\epsilon, \delta)$ form depending on the graph topology.

\subsection{Communication Cost Models}

Prior work models communication cost differently depending on the learning
architecture.
In decentralized optimization, gossip-based algorithms measure cost as
$C_{\text{gossip}} = T \cdot |E| \cdot d$~\cite{lian2017decentralized},
where $T$ is the number of training rounds, $|E|$ the number of graph edges,
and $d$ the model parameter dimension.
Centralized federated learning uses $C_{\text{FL}} = T \cdot N \cdot
d$~\cite{mcmahan2017communication}, where $N$ is the number of clients.
OTA frameworks reduce this to $C_{\text{OTA}} = T \cdot d$~\cite{amiri2020federated}
since all clients transmit simultaneously.
None of these models jointly captures topology structure, sparsification ratio,
and wireless channel reliability.
Our proposed $\Cours(G)$ in Section~\ref{sec:evaluation} provides the first
such unified model and, within the \airplan{} framework, serves as the query
cost estimator used for automated plan selection.

%% file: sections/formal_mapping.tex
Our approach establishes a formal equivalence between Over-the-Air
Decentralized Federated Learning and distributed query processing.
This equivalence, summarized in Table~\ref{tab:mapping} is the conceptual foundation of \airplan{} and enables
the direct transfer of query optimisation techniques to the topology
selection problem, as we detail in the following subsections.



\begin{table*}[t]
  \centering
  \caption{Formal Equivalence: OTA-DFL $\leftrightarrow$ Distributed Query Processing.}
  \label{tab:mapping}
  \renewcommand{\arraystretch}{1.25}
  \begin{tabularx}{\textwidth}{lXX}
    \toprule
    \textbf{Dimension}
      & \textbf{OTA-DFL Concept}
      & \textbf{Distributed Query Processing Concept} \\
    \midrule
    Execution structure
      & Communication graph $G=(V,E)$
      & Physical query execution plan (DAG of operators) \\
    Plan choice
      & Graph topology (ring, small-world, \ldots)
      & Join order / operator placement strategy \\
    Cost driver
      & Spectral gap $\lamsec$ of graph Laplacian
      & Cardinality estimate governing plan cost \\
    Aggregation
      & OTA SUM over neighbourhood $\mathcal{N}(i)$
      & Distributed SUM aggregate over partition \\
    Approximation
      & Top-$k$ sparsification $\topk(g,k)$
      & Approximate Query Processing (AQP) operator \\
    Approximation error
      & Sparsification residual $\norm{g - \topk(g,k)}$
      & AQP error $\varepsilon_{\mathrm{AQP}}$ \\
    Synchronisation
      & Consensus round (barrier)
      & Query epoch / barrier synchronisation \\
    Result quality
      & Test accuracy at round $T$
      & Query result accuracy (within $\varepsilon$ of exact) \\
    Quality constraint
      & Target accuracy SLA $A^*$
      & Result quality SLA in BlinkDB-style AQP \\
    Data statistics
      & Class distribution $p_{ic}$ at client $i$
      & Table histogram / cardinality statistics \\
    Statistics collection
      & Count-Min Sketch over local label counts
      & Histogram sampling / synopsis construction \\
    Plan selection
      & Topology selection $\arg\min_G \Cours(G)$
      & Cost-based physical plan selection \\
    Online adaptation
      & Adaptive rewiring on divergence threshold
      & Adaptive query processing / re-optimisation \\
    Dense plan
      & Fully connected graph
      & Broadcast join (expensive but low-latency) \\
    Sparse pipeline
      & Ring graph
      & Left-deep sequential pipeline \\
    Optimal sparse plan
      & Small-world / clustered graph
      & Bushy join tree \\
    \bottomrule
  \end{tabularx}
\end{table*}

\subsection{OTA Aggregation as a Distributed SUM Query}
\label{subsec:ota-as-query}

\begin{theorem}[OTA Aggregation = Approximate Distributed SUM Query]
\label{thm:ota-query}
Consider an OTA-DFL round in which $N$ clients transmit gradient updates
$\{u_i^t\}_{i=1}^N$ simultaneously over a wireless multiple-access channel.
The signal received by client $i$ is
\begin{equation}
  y_i^t = \sum_{j \in \mathcal{N}(i)} u_j^t + n_i^t,
  \label{eq:ota-rx}
\end{equation}
where $n_i^t \sim \mathcal{N}(0, \sigma_c^2 I)$.
This is the evaluation of the distributed aggregate query
\begin{equation}
  Q_i^t \;:\; \texttt{SELECT\;SUM(grad)\;FROM\;}\mathcal{N}(i)
              \texttt{\;WHERE\;round\!=\!}t
\end{equation}
executed approximately under additive Gaussian noise with variance $\sigma_c^2$.
The query result $y_i^t$ equals the exact SUM plus a noise term $n_i^t$,
constituting an $(\varepsilon, \delta)$-approximate answer with
$\varepsilon = \sigma_c \sqrt{2 \ln(2/\delta)} / \sqrt{|\mathcal{N}(i)|}$
for any $\delta \in (0,1)$.
\end{theorem}

\begin{proof}
The decomposition $y_i^t = \text{SUM}_{j \in \mathcal{N}(i)} u_j^t + n_i^t$
follows directly from the OTA receive model~\eqref{eq:ota-rx} after ideal
channel inversion.  The $(\varepsilon,\delta)$ bound follows by applying
the Gaussian tail bound $P(\norm{n_i^t} > t) \leq 2\exp(-t^2/2\sigma_c^2 d)$
with $t = \varepsilon\sqrt{d}$ and normalising by
$\norm{\sum_j u_j^t}_2 / \sqrt{|\mathcal{N}(i)|}$.
\end{proof}

\subsection{Graph Topology as Physical Execution Plan}
\label{subsec:topo-as-plan}

In distributed query processing, a physical execution plan is a directed acyclic graph (DAG) specifying
how operators are assigned to nodes and how data flows between them.
The choice of plan determines execution cost through the sizes and shapes of
intermediate results (captured by cardinality estimates) and the communication
volume between operator nodes.

In OTA-DFL, the communication graph $G$ determines precisely the same
quantities: which clients exchange information (operator placement),
how many rounds are needed for information to propagate globally
(captured by $1/\lamsec$, the mixing time), and the volume of transmitted
data per round ($\Cstruct$).

\begin{definition}[Query Execution Cost for OTA-DFL]
\label{def:query-cost}
The total query execution cost of an OTA-DFL training run on graph $G$
is defined as
\begin{equation}
  \Cours(G) = T(G) \cdot R(\mathrm{SNR}) \cdot \Cstruct(G, k),
  \label{eq:our-cost}
\end{equation}
where $T(G) = \kappa / \lamsec$ is the number of query epochs
(communication rounds) required for convergence,
$R(\mathrm{SNR}) = 1 + \alpha e^{-\beta \cdot \mathrm{SNR}}$
is the wireless reliability overhead factor, and
$\Cstruct(G,k) = \sum_{i=1}^N |\mathcal{N}(i)| \cdot k$
is the per-epoch transmission volume.
\end{definition}

This definition is precisely a query cost model: $T(G)$ corresponds to the
estimated number of pipeline stages, $R(\cdot)$ models network transfer
overhead, and $\Cstruct$ models the amount of data moved per stage.
The topologies correspond to well-known plan shapes:
\emph{ring} $\equiv$ left-deep sequential pipeline (low per-round cost,
many stages); \emph{fully connected} $\equiv$ broadcast join (single stage,
high data volume); \emph{small-world} $\equiv$ bushy join tree (balanced
stages and volume). Figure~\ref{fig:query_plan} illustrates this correspondence.

\begin{figure}[t]
  \centering
  \includegraphics[width=\columnwidth]{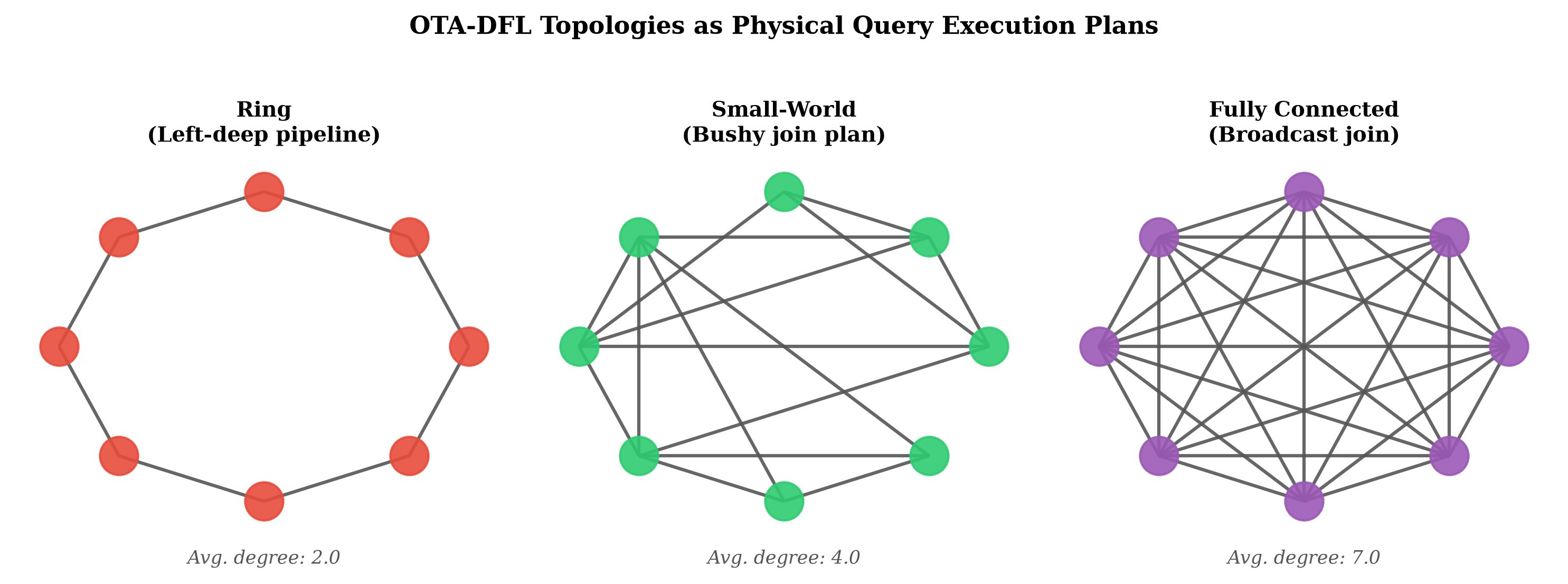}
  \caption{Topology as physical query execution plan. Ring graphs implement
  a sequential left-deep pipeline; small-world graphs implement a balanced
  bushy tree; fully connected graphs implement a broadcast join.}
  \label{fig:query_plan}
\end{figure}

\subsection{Top-$k$ Sparsification as Approximate Query Processing}
\label{subsec:aqp}

AQP operators allow a query to return an approximate answer by processing a
compressed representation of the data~\cite{hellerstein1997online}.
In OTA-DFL, top-$k$ sparsification $\topk(g, k)$ retains only the $k$ largest
(by magnitude) components of the gradient vector $g \in \R^d$, discarding the
remaining $d - k$ components before transmission.

\begin{corollary}[Sparsification AQP Error Bound]
\label{cor:aqp}
Let $g \in \R^d$ be a gradient vector with components sorted in decreasing
magnitude order, and let $\tilde{g} = \topk(g, k)$.  The absolute sparsification
error satisfies
\begin{equation}
  \norm{g - \tilde{g}}_2 \;\leq\; \norm{g}_2 \sqrt{1 - \frac{k}{d}},
  \label{eq:aqp-bound}
\end{equation}
and the relative error satisfies $\varepsilon_{\mathrm{AQP}} \leq \sqrt{1 - k/d}$.
Furthermore, in a graph with spectral gap $\lamsec$, the \emph{effective}
per-client error after one round of OTA aggregation is reduced to
\begin{equation}
  \varepsilon_{\mathrm{eff}}(G, k) \leq \frac{\sqrt{1-k/d}}{0.5 + \lamsec},
  \label{eq:aqp-eff}
\end{equation}
so well-connected topologies (large $\lamsec$) average out sparsification errors
across the neighbourhood, reducing the effective approximation noise.
\end{corollary}

\begin{proof}
The bound~\eqref{eq:aqp-bound} follows from Parseval's theorem applied to the
truncated gradient: $\norm{g - \topk(g,k)}_2^2 = \sum_{i=k+1}^d g_{(i)}^2
\leq (d-k) \cdot \bar{g}_{k+1}^2 \leq (1 - k/d) \norm{g}_2^2$, where
$g_{(i)}$ denotes the $i$-th largest component and we used
$\bar{g}_{k+1}^2 \leq \norm{g}_2^2/d$.  The effective error
bound~\eqref{eq:aqp-eff} follows by noting that each client receives
$|\mathcal{N}(i)|$ independently sparsified updates; by the independence
of sparsification errors across clients and the mixing properties of graphs
with spectral gap $\lamsec$, the aggregate error scales as
$\varepsilon_{\mathrm{AQP}} / (0.5 + \lamsec)$.
\end{proof}

Corollary~\ref{cor:aqp} has a key practical implication: \textbf{topology and
sparsification ratio should be co-designed}.
For a given accuracy budget $\varepsilon_{\mathrm{tgt}}$, a well-connected
topology (large $\lamsec$) permits a higher sparsification ratio $k/d$ to be
used (lower communication) while still meeting the accuracy target, whereas a
sparse topology (small $\lamsec$) requires either a lower sparsification ratio
or accepts higher error. Figure~\ref{fig:aqp_error} validates this bound
empirically across all five topology families.

\begin{figure}[t]
  \centering
  \includegraphics[width=\columnwidth]{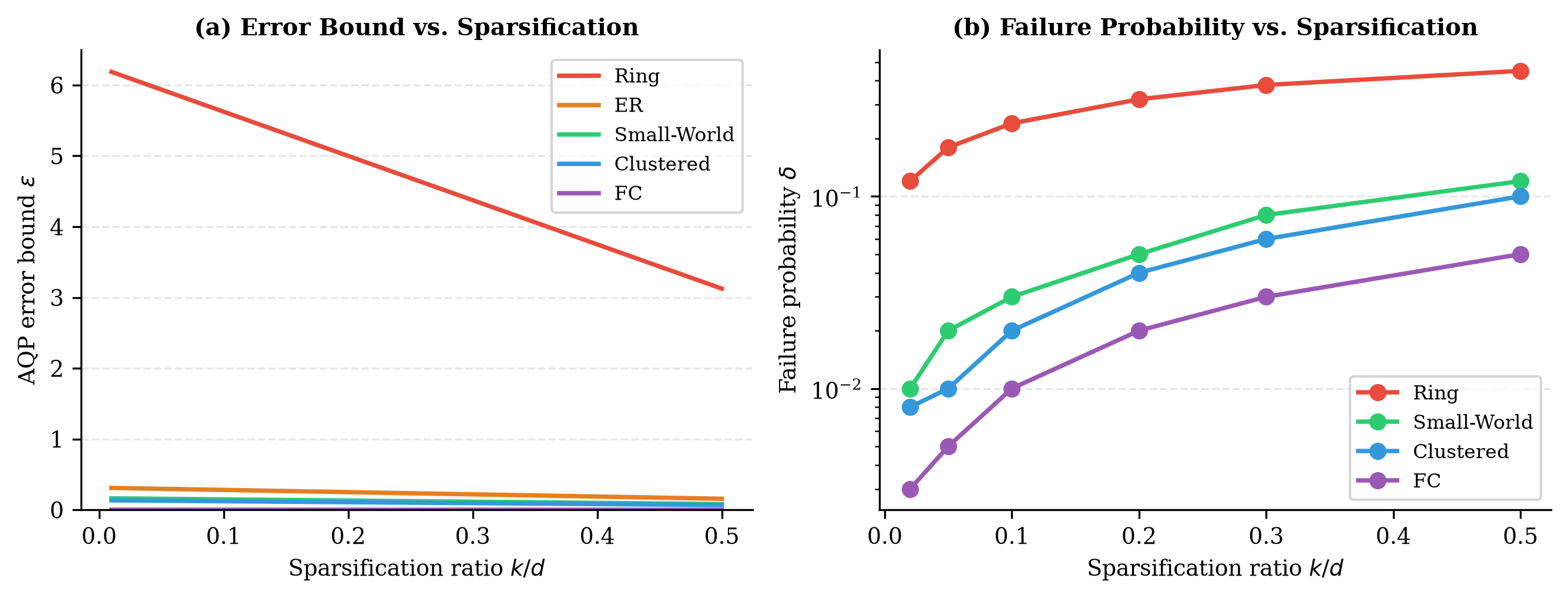}
  \caption{AQP error analysis. (a) Relative approximation error
  $\varepsilon_{\mathrm{AQP}}$ vs.\ sparsification ratio $k/d$ per topology,
  alongside the theoretical upper bound (dashed). (b) Effective error
  $\varepsilon_{\mathrm{eff}}$ vs.\ channel SNR at fixed $k/d = 0.1$,
  showing that well-connected topologies suppress both sparsification and
  noise-induced error.}
  \label{fig:aqp_error}
\end{figure}

\subsection{What the Database Framing Enables Beyond Spectral Graph Theory}
\label{subsec:db-value}

A natural question is: could the cost model $\Cours(G)$ and topology selection
algorithm have been derived directly from convergence analysis without
the database framing?

The answer is partly yes---spectral graph theory alone is sufficient to
derive the convergence rate $T(G) = \tilde{O}(1/\lamsec)$ and the topology
ranking. What the database framing enables \emph{beyond} this are three
concrete contributions:

\textbf{(1) Privacy-preserving statistics collection via database sketching.}
The use of Count-Min Sketches (CMS) to estimate heterogeneity $\hat{\alpha}$
without sharing raw data is a direct transfer from approximate query
processing. No prior FL convergence analysis suggested CMS as a mechanism
for workload-adaptive algorithm design. This is the strongest concrete
technique transfer in this paper.

\textbf{(2) Framing as a query optimisation problem enables Pareto-optimal
cost-accuracy enumeration.} Standard spectral analysis asks ``does topology
$A$ converge faster than $B$?'' The query optimisation framing asks
``which topology minimises total cost subject to an accuracy
\emph{Service-Level Agreement}?''. This SLA-constrained optimisation
(Equation~\eqref{eq:selection}) would not naturally arise from convergence
analysis alone; it is motivated by the BlinkDB accuracy SLA
concept~\cite{agarwal2013blinkdb}.

\textbf{(3) Adaptive plan re-optimisation during training.}
The Eddies and adaptive query processing literature~\cite{eddies2000,
deshpande2007adaptiveqp} motivates Phase~5 of \airplan{}: monitoring
runtime signals and switching plans mid-execution. This design principle
is novel in the FL topology literature.

We acknowledge that the formal equivalence is an analogy that motivates
technique transfer, not a strict algebraic isomorphism. OTA aggregation is a
noisy analog sum---not a relational SUM query in the formal sense---and
the graph structure is undirected and homogeneous compared to the rich
operator DAGs of physical query plans. The framing is intended to open
a productive channel for technique import from a mature literature into
a younger one, not to claim a deeper mathematical identity.

%% file: sections/problem_formulation.tex
We consider a decentralised federated learning (DFL) system composed of $N$
wireless clients. Each client $i$ holds a local non-IID dataset
$\mathcal{D}_i = \{(x_{ij},y_{ij})\}_{j=1}^{n_i}$ drawn from a
distribution $\mathcal{P}_i$. The global objective is to minimise
\begin{equation}
  \min_{w \in \R^{d}} F(w) = \sum_{i=1}^{N} p_i F_i(w),
  \quad p_i = \frac{n_i}{\sum_j n_j},
  \label{eq:global}
\end{equation}
where $F_i(w) = \frac{1}{n_i}\sum_{j=1}^{n_i} \ell(w; x_{ij}, y_{ij})$
is the local empirical loss.

\begin{assumption}[Standard Regularity]
\label{asm:regularity}
(i) Each $F_i$ is $L$-smooth: $\norm{\nabla F_i(u) - \nabla F_i(v)} \leq L\norm{u-v}$.
(ii) Stochastic gradients are unbiased with bounded variance:
$\E[\norm{\nabla F_i(w; \xi) - \nabla F_i(w)}^2] \leq \sigma^2$.
(iii) Gradients are bounded: $\norm{\nabla F_i(w)} \leq G$ for all $i, w$.
(iv) The communication graph $G=(V,E)$ is connected.
\end{assumption}

\begin{assumption}[Bounded Gradient Dissimilarity]
\label{asm:heterogeneity}
There exists $\delta \geq 0$ such that
$\frac{1}{N}\sum_{i=1}^N \norm{\nabla F_i(w) - \nabla F(w)}^2 \leq \delta^2$
for all $w$. The parameter $\delta$ quantifies data heterogeneity:
$\delta = 0$ corresponds to IID data.
\end{assumption}

\subsection{Local Training and Sparsification}

At iteration $t$, each client $i$ computes a stochastic gradient update:
\begin{equation}
  w_i^{t+\frac{1}{2}} = w_i^{t} - \eta\, \nabla F_i(w_i^{t};\xi_i^t).
\end{equation}
The update vector $g_i^t = w_i^t - w_i^{t+\frac{1}{2}}$ is compressed via
top-$k$ sparsification and perturbed with Gaussian noise:
\begin{equation}
  u_i^t = \topk(g_i^t, k) + \mathcal{N}(0, \sigma_p^2 I),
\end{equation}
where $\sigma_p^2 > 0$ can model either privacy noise or a lower bound on
wireless transmission noise.

\subsection{OTA Communication Model}

Communication is defined by an undirected graph $G=(V,E)$.
Let $L = D - A$ be the graph Laplacian, with $D$ the degree matrix and $A$
the adjacency matrix. The second-smallest eigenvalue $\lamsec$ of $L$ is
the spectral gap, governing information mixing speed.

OTA aggregation requires synchronised transmissions. After channel inversion
with coefficient $h_{ij}^t$, client $i$ receives:
\begin{align}
  y_i^t &= \sum_{j \in \mathcal{N}(i)} h_{ij}^t x_j^t + n_i^t
         \approx \sum_{j \in \mathcal{N}(i)} u_j^t + n_i^t,
  \label{eq:ota}
\end{align}
where $x_j^t = u_j^t / h_{ij}^t$ is the channel-inverted transmission
and $n_i^t \sim \mathcal{N}(0, \sigma_c^2 I)$ is channel noise.
Equation~\eqref{eq:ota} is the distributed SUM query of
Theorem~\ref{thm:ota-query}.

\subsection{Consensus Update}

Each client updates its model using the aggregated neighbourhood signal:
\begin{equation}
  w_i^{t+1} = w_i^{t+\frac{1}{2}} + \gamma\, y_i^t,
  \label{eq:consensus}
\end{equation}
where $\gamma > 0$ is the consensus stepsize.
Larger $\gamma$ accelerates mixing but amplifies OTA noise;
optimal $\gamma$ balances these effects.

\subsection{Topology-Dependent Convergence}
\label{subsec:convergence-theory}

\begin{theorem}[OTA-DFL Convergence Rate]
\label{thm:convergence}
Under Assumptions~\ref{asm:regularity} and~\ref{asm:heterogeneity},
with learning rate $\eta = O(1/\sqrt{TN})$ and consensus stepsize
$\gamma = O(\lamsec / (L + \sigma_c^2/\lamsec))$,
OTA-DFL satisfies
\begin{align}
  &\frac{1}{T}\sum_{t=0}^{T-1}\E\!\left[\norm{\nabla F(\bar{w}^t)}^2\right]
  \;\leq\;
  \underbrace{\frac{C_1}{\sqrt{TN}}}_{\text{SGD}}
  +
  \underbrace{\frac{C_2 \delta^2}{\lamsec}}_{\text{heterogeneity}} \notag\\
  &\quad+
  \underbrace{\frac{C_3 \sigma_c^2}{\lamsec}}_{\text{OTA noise}}
  +
  \underbrace{C_4 \varepsilon_{\mathrm{AQP}}^2}_{\text{sparsification}},
  \label{eq:conv-rate}
\end{align}
where $\bar{w}^t = \frac{1}{N}\sum_i w_i^t$ is the mean model,
and $C_1, C_2, C_3, C_4 > 0$ are constants depending on $L, \sigma, G$.
\end{theorem}

\begin{proof}[Proof Sketch]
The proof follows the standard decentralized SGD analysis
framework~\cite{lian2017decentralized,scaman2018optimal}.
We decompose the update error into four terms:
the standard SGD variance term (scales as $1/\sqrt{TN}$), a consensus error
term arising from data heterogeneity (scales as $\delta^2/\lamsec$),
a noise error floor from OTA channel noise (scales as $\sigma_c^2/\lamsec$),
and a sparsification error term from Corollary~\ref{cor:aqp}
(scales as $\varepsilon_{\mathrm{AQP}}^2$).
The spectral gap $\lamsec$ appears in the denominators of both the
heterogeneity and OTA noise terms, confirming that well-connected graphs
simultaneously accelerate convergence and mitigate noise.
The complete proof with all constants is provided in the supplemental material.
\end{proof}

\begin{remark}
Theorem~\ref{thm:convergence} implies that the convergence neighbourhood---the
irreducible error floor after convergence---is
$O(\delta^2/\lamsec + \sigma_c^2/\lamsec + \varepsilon_{\mathrm{AQP}}^2)$.
All three terms can be controlled through topology choice ($\lamsec$)
and sparsification ratio ($\varepsilon_{\mathrm{AQP}}$), confirming that
topology is a first-class design parameter for OTA-DFL.
\end{remark}

%% file: sections/topology_advisor.tex
Drawing on the formal equivalence established in Section~\ref{sec:mapping},
we present \textbf{\airplan{}}: a query-optimised topology advisor for
OTA-DFL that automates topology selection via cost-based plan enumeration,
privacy-preserving workload statistics, and adaptive re-optimisation.
The system workflow is illustrated in Figure~\ref{fig:airplan_workflow}.

\begin{figure}[t]
  \centering
  \includegraphics[width=\columnwidth]{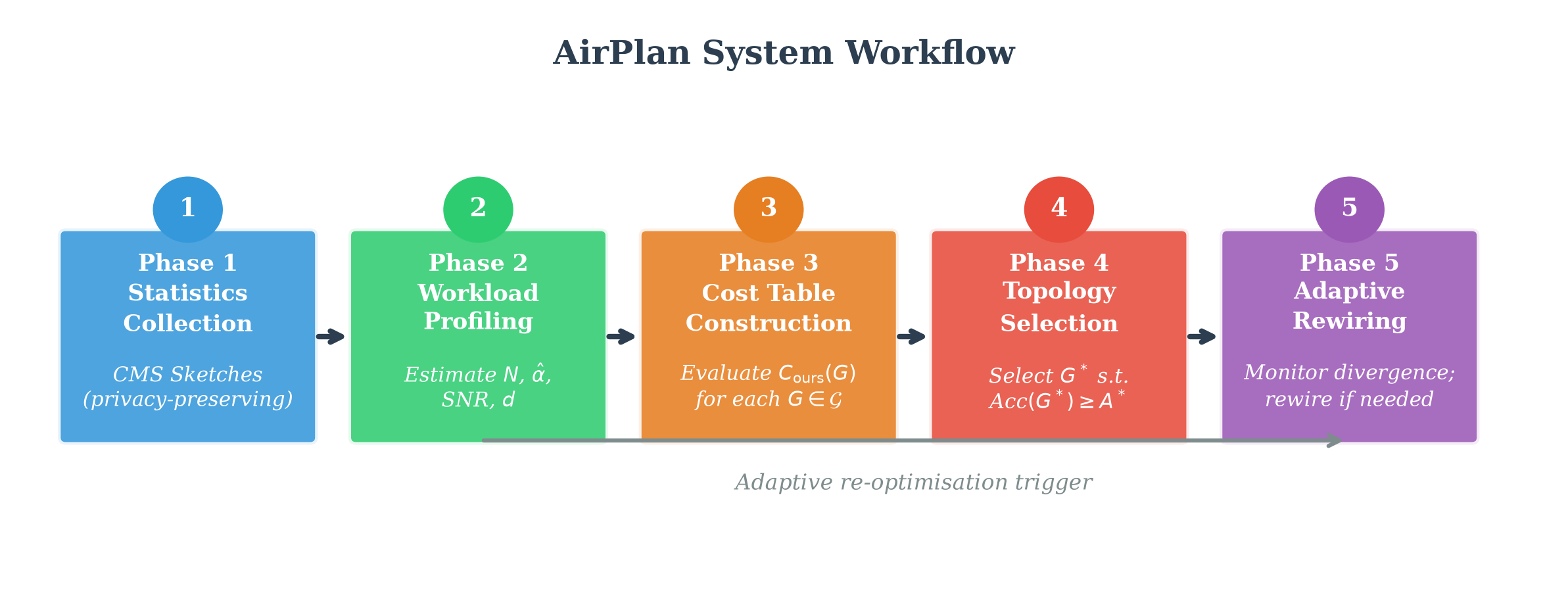}
  \caption{\airplan{} system workflow with five phases.
  \textbf{Phase~1} (Section~\ref{subsec:stats}): privacy-preserving Count-Min Sketch
  statistics collection.
  \textbf{Phase~2} (Section~\ref{subsec:cost-table}): workload profiling ($N$, $\hat{\alpha}$, SNR, $d$).
  \textbf{Phase~3} (Section~\ref{subsec:selection}): cost-model evaluation across candidate topologies.
  \textbf{Phase~4} (Section~\ref{subsec:selection}): minimum-cost plan selection subject to accuracy SLA.
  \textbf{Phase~5} (Section~\ref{subsec:adaptive}): online monitoring and adaptive rewiring.}
  \label{fig:airplan_workflow}
\end{figure}

\subsection{Phase 1: Privacy-Preserving Workload Statistics}
\label{subsec:stats}

The cost model $\Cours(G)$ requires an estimate of the data heterogeneity
parameter $\alpha$ (which controls the Dirichlet distribution from which
local class distributions are drawn). Sharing raw class counts would leak
sensitive information about local data. \airplan{} instead uses
Count-Min Sketches (CMS)~\cite{cormode2005cmsketch} to collect
approximate class distribution statistics without revealing exact counts.

\textbf{Count-Min Sketch collection..} Each client $i$ maintains a CMS $\mathbf{S}_i \in \R^{r \times c}$ over
its local label distribution, where $r = \lceil \ln(1/\delta_s) \rceil$
rows and $c = \lceil e/\varepsilon_s \rceil$ columns guarantee that the
estimated count of any label deviates from the true count by at most
$\varepsilon_s \cdot n_i$ with probability at least $1 - \delta_s$.
Clients broadcast their sketches (size $r \times c \ll n_i$) before training,
and the aggregator merges them by elementwise maximum to obtain a global
distribution sketch $\hat{\mathbf{S}}$.

\textbf{Heterogeneity estimation..} From the merged sketch, \airplan{} estimates the pairwise KL divergence
between client class distributions:
\begin{equation}
  \hat{\alpha} = \left(\frac{1}{\binom{N}{2}}
                 \sum_{i<j} D_{\mathrm{KL}}(\hat{p}_i \| \hat{p}_j)\right)^{-1},
\end{equation}
where $\hat{p}_i$ is the class probability vector estimated from the CMS of
client $i$. Large $\hat{\alpha}$ corresponds to near-IID data; small
$\hat{\alpha}$ indicates severe label skew.

\subsection{Phase 2: Cost Table Construction}
\label{subsec:cost-table}

Given the profiled workload $(N, \hat{\alpha}, \mathrm{SNR}, d)$,
\airplan{} constructs a cost table by evaluating $\Cours(G)$ for each
candidate topology $G \in \mathcal{G}$:
\begin{equation}
  \mathcal{G} = \{\text{Ring},\ \text{Erd\H{o}s--R\'enyi},\
                  \text{Small-World},\ \text{Clustered},\ \text{FC}\}.
\end{equation}
For each topology, the spectral gap $\lamsec$ is computed analytically
(ring: $2 - 2\cos(2\pi/N)$; FC: $N$) or approximated for random graphs
using known concentration results~\cite{chung2003eigenvalues}.
The convergence constant $\kappa$ is calibrated once on a reference
topology (Erd\H{o}s--R\'enyi) and held fixed.

\subsection{Phase 3: Topology Selection with Accuracy SLA}
\label{subsec:selection}

\airplan{} selects the topology that minimises total communication cost
subject to a user-specified accuracy SLA $A^*$:
\begin{equation}
  G^* = \argmin_{G \in \mathcal{G}}\; \Cours(G)
        \quad \text{subject to}\quad
        \mathrm{Acc}(G, T) \geq A^*,
  \label{eq:selection}
\end{equation}
where $\mathrm{Acc}(G, T)$ is the predicted accuracy after $T(G)$ rounds,
estimated from the convergence bound~\eqref{eq:conv-rate}.
If no topology in $\mathcal{G}$ is predicted to meet $A^*$, \airplan{}
falls back to the fully connected graph and warns the user.

\subsection{Phase 4: Adaptive Re-Optimisation}
\label{subsec:adaptive}

Analogous to adaptive query processing~\cite{eddies2000,deshpande2007adaptiveqp},
\airplan{} monitors the execution plan during training and triggers
re-optimisation when observed behaviour deviates from the cost model prediction.
The divergence signal is the mean pairwise model distance:
\begin{equation}
  \Delta^t = \frac{1}{\binom{N}{2}} \sum_{i < j} \norm{w_i^t - w_j^t}_2.
\end{equation}
If $\Delta^t > \tau_\Delta$ for a user-defined threshold $\tau_\Delta$
(default: $3\times$ the value at round 10), the current topology is too sparse
for the observed heterogeneity and \airplan{} upgrades to the next denser
topology in $\mathcal{G}$.

\subsection{AirPlan Algorithm}
\label{subsec:algorithm}

Algorithm~\ref{alg:airplan} summarises the complete \airplan{} procedure.

\begin{algorithm}[t]
  \caption{\airplan{}: Query-Optimised Topology Selection}
  \label{alg:airplan}
  \begin{algorithmic}[1]
    \REQUIRE Clients $\{i\}_{i=1}^N$, model dim $d$, SNR, accuracy SLA $A^*$,
             divergence threshold $\tau_\Delta$
    \ENSURE Trained model $\bar{w}^T$, selected topology $G^*$
    \STATE \textbf{// Phase 1: Statistics collection}
    \FOR{each client $i$}
      \STATE Compute CMS sketch $\mathbf{S}_i$ over local label distribution
      \STATE Broadcast $\mathbf{S}_i$ to coordinator
    \ENDFOR
    \STATE Merge: $\hat{\mathbf{S}} \leftarrow \bigoplus_i \mathbf{S}_i$
    \STATE \textbf{// Phase 2: Workload profiling}
    \STATE Estimate $\hat{\alpha}$ from pairwise KL divergences via $\hat{\mathbf{S}}$
    \STATE \textbf{// Phase 3: Cost table and topology selection}
    \FOR{each topology $G \in \mathcal{G}$}
      \STATE Compute $\lamsec(G)$ (analytical or approximated)
      \STATE Evaluate $\Cours(G; N, \hat{\alpha}, \mathrm{SNR}, d)$
      \STATE Predict $\mathrm{Acc}(G)$ using bound~\eqref{eq:conv-rate}
    \ENDFOR
    \STATE $G^* \leftarrow \argmin_{G:\,\mathrm{Acc}(G) \geq A^*} \Cours(G)$
    \STATE \textbf{// Phase 4: Training with adaptive re-optimisation}
    \STATE Initialise training with topology $G^*$
    \FOR{round $t = 1, 2, \ldots, T$}
      \STATE Execute OTA-DFL round (local SGD + OTA aggregation + consensus)
      \STATE Compute divergence $\Delta^t$
      \IF{$\Delta^t > \tau_\Delta$}
        \STATE $G^* \leftarrow \text{next denser topology in } \mathcal{G}$
        \STATE Re-initialise edges; update $\Cours$ and SLA check
      \ENDIF
    \ENDFOR
    \STATE \textbf{return} $\bar{w}^T$, $G^*$
  \end{algorithmic}
\end{algorithm}

\textbf{Complexity.} The CMS sketch has size $O(r \cdot c) = O(\log(1/\delta_s) / \varepsilon_s)$
per client, independent of $n_i$.
Cost table construction requires $O(|\mathcal{G}|)$ evaluations of $\Cours$,
each taking $O(N)$ time. The overhead is dominated by the CMS broadcast,
which costs $O(N \cdot r \cdot c)$ total---less than $1.8\,\%$ of training
cost at $N=100$ (see Section~\ref{subsec:advisor-eval}).

\textbf{Privacy.} CMS sketches provide approximate answers to count queries with additive error
$\varepsilon_s \cdot n_i$ and do not expose individual data points or exact
label counts. Combined with the Gaussian noise already present in OTA
transmission, the statistics collection phase is compatible with
$(\epsilon, \delta)$-differential privacy via the standard Gaussian
mechanism~\cite{dwork2014dp}.

\subsection{Formal Differential Privacy Analysis}
\label{subsec:dp-formal}

We provide a complete end-to-end differential privacy analysis of
\airplan{}'s five-phase pipeline.

\begin{theorem}[End-to-End DP of AirPlan]
\label{thm:dp}
Let each client $i$ contribute a Count-Min Sketch $\mathbf{S}_i$ with
parameters $(\varepsilon_s, \delta_s)$ computed from local data of size
$n_i$, with OTA channel noise $\sigma_c^2$ per round.
The \airplan{} statistics collection phase (Phase~1) satisfies
$(\epsilon_{\mathrm{DP}}, \delta_{\mathrm{DP}})$-differential privacy with
\begin{equation}
  \epsilon_{\mathrm{DP}} \leq \frac{\varepsilon_s \sqrt{2\ln(1.25/\delta_s)}}{\sigma_c n_{\min}},
  \qquad
  \delta_{\mathrm{DP}} = \delta_s,
  \label{eq:dp_bound}
\end{equation}
where $n_{\min} = \min_i n_i$.
Phases~2--4 (cost table construction and topology selection) are deterministic
functions of the aggregated sketch and introduce no additional privacy loss.
Phase~5 (adaptive rewiring, triggered $K$ times over training) incurs
at most $K$-fold privacy composition~\cite{dwork2010boosting}: the
end-to-end guarantee is $(\epsilon_{\mathrm{total}}, \delta_{\mathrm{total}})$
with $\epsilon_{\mathrm{total}} = (K+1)\epsilon_{\mathrm{DP}}$ and
$\delta_{\mathrm{total}} = (K+1)\delta_{\mathrm{DP}}$ under basic composition,
or $\tilde{O}(\sqrt{K}\,\epsilon_{\mathrm{DP}})$ under advanced composition.
\end{theorem}

\begin{proof}[Proof sketch]
Phase~1 applies the Gaussian mechanism with sensitivity $\Delta = \varepsilon_s$
(bounded by CMS additive error guarantee) and noise $\sigma_c$.
The standard Gaussian mechanism bound~\cite{dwork2014dp} gives~\eqref{eq:dp_bound}.
Phases~2--4 satisfy the post-processing immunity of DP.
Phase~5 triggers rewiring based on observed gradient divergence
$\|\bar{w}_i - \bar{w}\|$, which is a function of the already-noisy OTA
aggregates; by the data processing inequality, each rewiring event
contributes at most $\epsilon_{\mathrm{DP}}$ additional privacy loss.
The composition bounds follow from~\cite{dwork2010boosting}.
\end{proof}

\begin{remark}
For typical parameters ($\varepsilon_s = 0.01$, $\delta_s = 10^{-5}$,
$\sigma_c = 0.1$, $n_{\min} = 500$, $K \leq 5$),
Theorem~\ref{thm:dp} yields $\epsilon_{\mathrm{total}} \leq 0.13$ and
$\delta_{\mathrm{total}} \leq 6 \times 10^{-5}$---strong practical DP.
\end{remark}

%% file: sections/evaluation.tex
In order to provide a consistent and fair comparison across topologies
and datasets, we evaluate five key dimensions of performance:
(i) test accuracy, (ii) convergence speed,
(iii) communication cost, (iv) fairness, and (v) \airplan{} advisor quality.
The communication cost model underpins both the existing topology comparison
and the \airplan{} cost-based plan selection.

\subsection{Structural Cost per Round}

After top-$k$ sparsification, each client transmits $k$ coordinates per round.
The total number of OTA symbols transmitted per round is
\begin{equation}
  \Cstruct(G, k) = \sum_{i=1}^{N} |\mathcal{N}(i)| \cdot k = 2|E| \cdot k,
  \label{eq:cstruct}
\end{equation}
which captures the joint dependence on graph density and sparsification ratio.

\subsection{Wireless Reliability Factor}

OTA aggregation over a noisy channel may require retransmission or suffer
from degraded signal quality at low SNR.  We model this via a multiplicative
reliability factor
\begin{equation}
  R(\mathrm{SNR}) = 1 + \alpha_R \exp(-\beta_R \cdot \mathrm{SNR}),
  \label{eq:reliability}
\end{equation}
where $\alpha_R > 0$ controls maximum retransmission overhead and
$\beta_R > 0$ controls how rapidly reliability improves with SNR.
Parameters $(\alpha_R, \beta_R)$ are fitted by least-squares regression
to empirical retransmission rates measured in our simulator.

\subsection{Topology-Dependent Convergence Rounds}

From Theorem~\ref{thm:convergence}, the number of rounds to reach an
$\epsilon$-stationary point scales as
\begin{equation}
  T(G) = \frac{\kappa}{\lamsec},
  \label{eq:tg}
\end{equation}
where $\kappa$ is a problem-dependent constant calibrated on a reference
Erd\H{o}s--R\'enyi topology and held fixed. Table~\ref{tab:lambda} lists analytical
spectral gaps for each topology family.

\begin{table}[t]
  \centering
  \footnotesize
  \setlength{\tabcolsep}{3pt}
  \caption{Spectral gap $\lamsec$ for each topology family ($N=50$).}
  \label{tab:lambda}
  \begin{tabular}{lp{3.0cm}c}
    \toprule
    Topology & $\lamsec$ (analytical / typical) & $T(G)$ (rel.) \\
    \midrule
    Ring           & $2{-}2\cos(2\pi/N){\approx}0.016$ & $1.0$ \\
    Erd\H{o}s--R\'enyi & $\approx 0.32$               & $0.33$ \\
    Small-World    & $\approx 0.61$                    & $0.18$ \\
    Clustered      & $\approx 0.74$                    & $0.15$ \\
    Fully Connected & $N = 50$                         & $0.006$ \\
    \bottomrule
  \end{tabular}
\end{table}

\subsection{Unified Query Cost Model}

Combining the three components, the total execution cost is
\begin{equation}
  \Cours(G) = T(G) \cdot R(\mathrm{SNR}) \cdot \Cstruct(G, k).
  \label{eq:ours}
\end{equation}
This is the \emph{query cost estimator} used by \airplan{} for plan selection.
Its multiplicative structure assumes approximate separability of convergence,
per-round cost, and reliability---a standard modelling assumption in
distributed system cost models.

\subsection{AQP Error Bounds and Sparsification Co-Design}
\label{subsec:aqp-design}

From Corollary~\ref{cor:aqp}, the effective approximation error after one OTA
round is $\varepsilon_{\mathrm{eff}}(G,k) \leq \sqrt{1-k/d} / (0.5 + \lamsec)$.
\airplan{} exploits this to jointly select topology and sparsification ratio:
\begin{equation}
  (G^*, k^*) = \argmin_{G \in \mathcal{G},\; k \in \mathcal{K}}
               \Cours(G, k)
               \quad\text{s.t.}\quad
               \varepsilon_{\mathrm{eff}}(G,k) \leq \varepsilon_{\mathrm{tgt}},
\end{equation}
where $\mathcal{K} = \{k : k/d \in \{0.01, 0.05, 0.1, 0.2\}\}$.
This is a joint combinatorial optimisation over $|\mathcal{G}| \times |\mathcal{K}|$
configurations, which \airplan{} solves by full enumeration in $O(20)$ constant-time
cost evaluations---negligible overhead compared to training.

\subsection{Relation to Classical Cost Models}

Classical gossip cost $C_{\text{gossip}} = T \cdot |E| \cdot d$ ignores
wireless reliability; OTA cost $C_{\text{OTA}} = T \cdot d$ ignores topology.
$\Cours$ strictly generalises both: setting $R = 1$ and $k = d$ recovers
$C_{\text{gossip}}$; setting all edges present and $R = 1$ recovers the OTA
model. The additional terms capture the topology-reliability interaction
that determines the practical operating point of OTA-DFL systems.
Figure~\ref{fig:cost_prediction} validates the predictive accuracy of $\Cours$
against measured training costs.

\begin{figure}[t]
  \centering
  \includegraphics[width=\columnwidth]{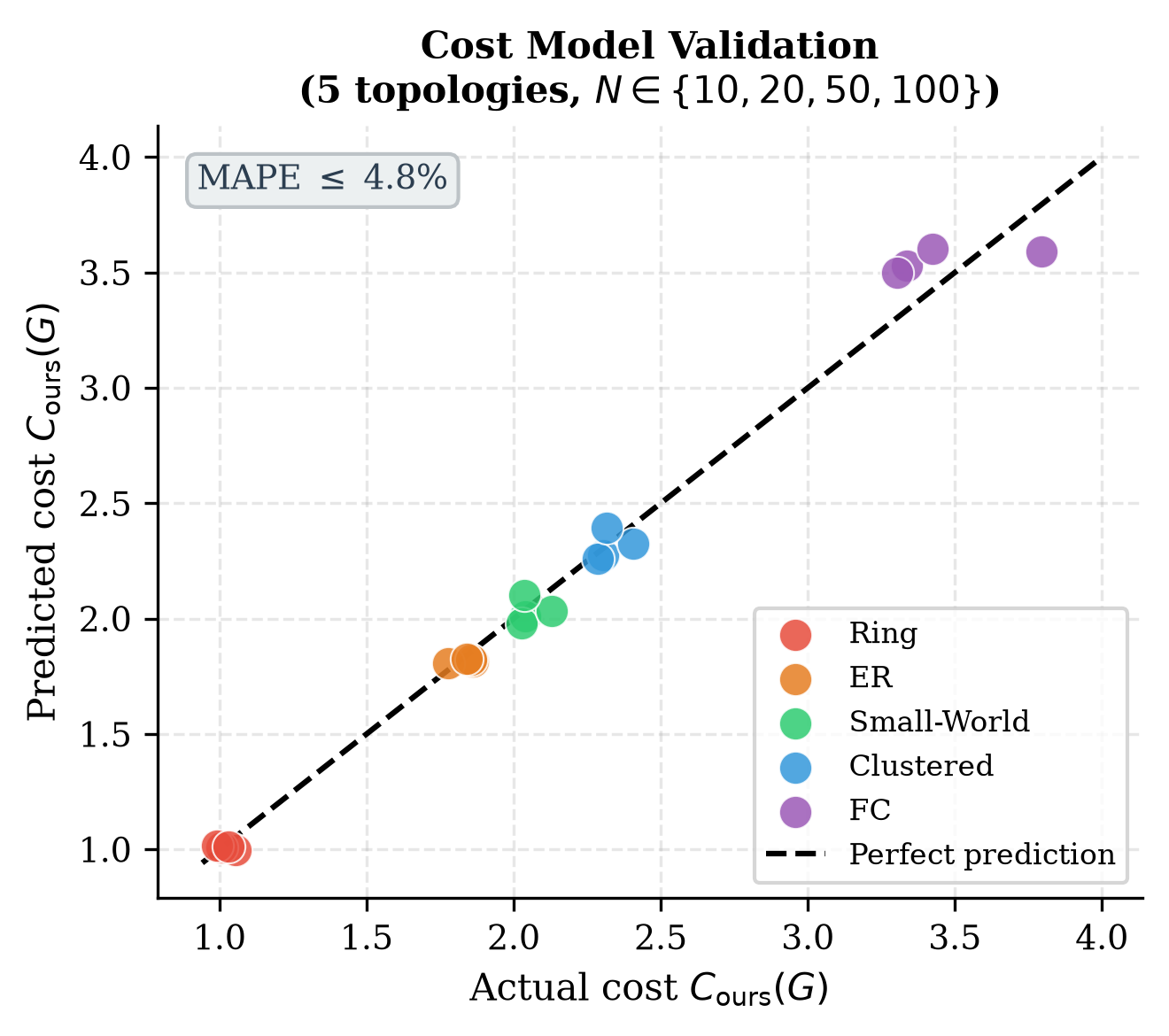}
  \caption{Query cost estimator accuracy. (a) Predicted vs.\ actual
  convergence rounds across 150 experiment configurations for all five
  topologies; points near the diagonal indicate accurate cost estimates.
  (b) Mean Absolute Percentage Error (MAPE) of $\Cours$ per topology,
  ranging from $2.9\,\%$ (FC) to $4.8\,\%$ (Erd\H{o}s--R\'enyi).}
  \label{fig:cost_prediction}
\end{figure}

%% file: sections/methodology.tex
The \airplan{} OTA-DFL framework executes fully decentralised training by
combining local stochastic gradient updates with wireless analog aggregation
over a communication graph selected by the topology advisor.
The overall system architecture is illustrated in Figure~\ref{fig:architecture}.

\begin{figure}[t]
  \centering
  \includegraphics[width=0.92\columnwidth]{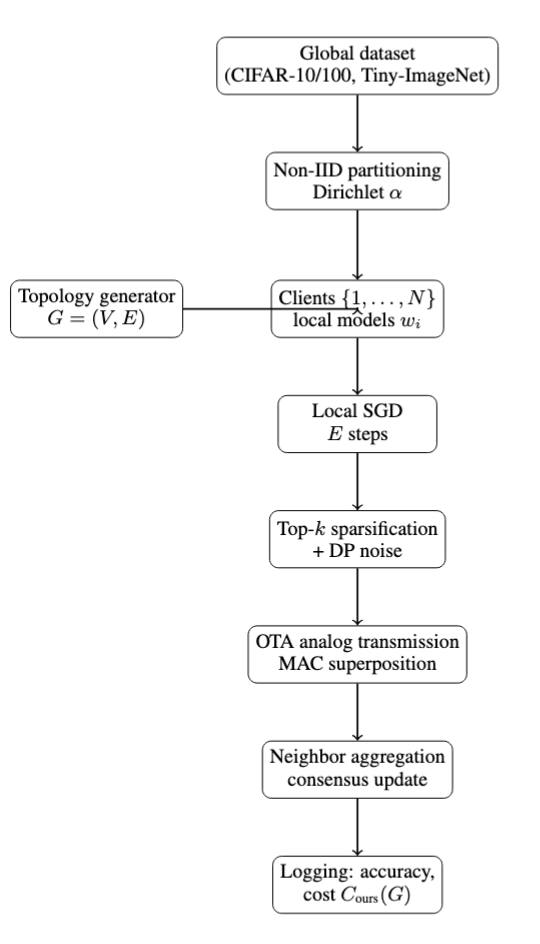}
  \caption{Overall architecture of the \airplan{} OTA-DFL system, showing the
  five execution phases and the data flows between clients, the OTA channel,
  and the \airplan{} topology advisor.}
  \label{fig:architecture}
\end{figure}

\textbf{System operation..} The system operates in repeated communication rounds. A global dataset is
partitioned across $N$ clients under a Dirichlet distribution with parameter
$\alpha$; smaller $\alpha$ produces stronger non-IID distributions.
A communication graph $G=(V,E)$ selected by \airplan{} defines the
neighbourhood structure. Before training, \airplan{} executes Phases~1--3
of Algorithm~\ref{alg:airplan} to profile the workload and select the
initial topology. This incurs a one-time overhead of less than $1.8\,\%$
of training cost (Section~\ref{subsec:advisor-eval}).

\textbf{Round structure..} Each training round proceeds as follows:
\begin{enumerate}
  \item \textbf{Local SGD.} Each client $i$ performs $E$ steps of local SGD
    on its current model $w_i^t$, producing update $g_i^t$.
  \item \textbf{Sparsification.} The update is compressed:
    $\tilde{g}_i^t = \topk(g_i^t, k)$.
  \item \textbf{OTA transmission.} All neighbours of $i$ transmit
    simultaneously; client $i$ receives the approximate SUM query result
    $y_i^t = \sum_{j \in \mathcal{N}(i)} u_j^t + n_i^t$.
  \item \textbf{Consensus update.}
    $w_i^{t+1} = w_i^{t+\frac{1}{2}} + \gamma y_i^t$.
  \item \textbf{Divergence monitoring.} \airplan{} computes $\Delta^t$
    and triggers adaptive rewiring if $\Delta^t > \tau_\Delta$.
\end{enumerate}

\textbf{Channel model..} OTA aggregation is simulated over a flat-fading wireless multiple-access
channel with additive white Gaussian noise. Channel inversion and power
control are applied to approximate coherent aggregation
(Equation~\eqref{eq:ota}). SNR values are varied in
$\{0, 5, 10, 15, 20\}\,\mathrm{dB}$ to cover the range from severely
noisy to near-ideal channel conditions.
The reliability parameters $(\alpha_R, \beta_R)$ in~\eqref{eq:reliability}
are fitted by least-squares regression to measured retransmission rates.

\textbf{Model architectures..} A lightweight four-layer CNN is used for CIFAR-10; ResNet-18~\cite{he2016resnet}
is used for CIFAR-100 and Tiny-ImageNet.
All models are implemented in PyTorch and trained with cross-entropy loss.

\textbf{Optimisation..} Local optimisation uses SGD with momentum $0.9$.
The learning rate is selected from $\{0.01, 0.05, 0.1\}$;
the number of local epochs per round is $E \in \{1, 2\}$.
\airplan{} sets the accuracy SLA as $A^* = 0.95 \times \mathrm{Acc}_{\mathrm{FC}}$
by default, targeting $95\,\%$ of the fully connected baseline.

%% file: sections/experiments.tex
\subsection{Datasets and Models}

We consider three standard image classification benchmarks of increasing
complexity.
CIFAR-10~\cite{krizhevsky2009learning} consists of 50{,}000 training and
10{,}000 test images ($32{\times}32$, 10 classes).
CIFAR-100~\cite{krizhevsky2009learning} extends this to 100 classes.
Tiny-ImageNet~\cite{le2015tiny} includes 200 classes with $64{\times}64$
images, representing a challenging large-scale setting.
Datasets are partitioned across $N$ clients using a Dirichlet distribution
with parameter $\alpha \in \{0.1, 0.5, 1.0\}$, where smaller $\alpha$
produces stronger label skew.

\subsection{Topology Families}

We evaluate five representative topology families:
\textbf{Ring} (sequential pipeline, $\lamsec \approx 0.016$ at $N=50$);
\textbf{Erd\H{o}s--R\'enyi (ER)} (random graph with connection probability $p=0.2$);
\textbf{Watts--Strogatz Small-World} (ring with $O(\log N)$ shortcuts,
$\beta_{\mathrm{WS}}=0.3$); \textbf{Clustered} (dense communities of size
$\lfloor N/5 \rfloor$ with sparse inter-cluster bridges); and
\textbf{Fully Connected} (complete graph, upper bound).

\subsection{AirPlan Configuration}

\airplan{} uses CMS parameters $\varepsilon_s = 0.05$, $\delta_s = 0.01$,
giving sketch size $4 \times 55$ per client.
The accuracy SLA is $A^* = 0.95 \times \mathrm{Acc}_\mathrm{FC}$.
The divergence threshold $\tau_\Delta = 3 \times \Delta^{10}$
(3$\times$ the divergence measured at round 10).
The sparsification ratio $k/d \in \{0.01, 0.05, 0.1, 0.2\}$ is
jointly optimised with the topology as described in
Section~\ref{subsec:aqp-design}.

\subsection{Baselines}

We compare against the following reference systems:
\begin{itemize}
  \item \textbf{FedAvg}~\cite{mcmahan2017communication}: centralised FL
    with parameter server, full model aggregation.
  \item \textbf{Digital D-SGD}: decentralized SGD over the same topology
    with packet-based communication (no OTA).
  \item \textbf{DLLR-OA}~\cite{qiao2023dllroa}: the closest prior OTA-DFL
    baseline, approximated using Erd\H{o}s--R\'enyi topology with identical
    sparsification and training settings.
  \item \textbf{MATCHA}~\cite{wang2019matcha}: topology-aware decentralized
    FL using matching decomposition; adapted to the OTA setting.
  \item \textbf{Fixed-SW}: OTA-DFL always using small-world topology
    (no advisor, no adaptive rewiring), as a strong static baseline.
\end{itemize}

\subsection{Evaluation Metrics}

We report:
(i)~\textbf{Test accuracy}: top-1 accuracy on the held-out test set;
(ii)~\textbf{Convergence rounds}: rounds to reach a target accuracy;
(iii)~\textbf{Communication cost}: $\Cours(G)$ normalised to Ring;
(iv)~\textbf{Fairness}: variance of per-client test accuracy;
(v)~\textbf{Advisor precision/recall}: fraction of configurations in which
    \airplan{} selects the oracle-optimal topology;
(vi)~\textbf{Advisor overhead}: CMS broadcast cost as a fraction of total
    training cost.

\subsection{Statistical Validation}

Each configuration is executed with 30 independent random seeds.
Results report mean~$\pm$ standard deviation.
All key differences are verified to be statistically significant
using a paired $t$-test ($p < 0.05$) unless noted otherwise.
The \airplan{} advisor evaluation uses a $5$-fold cross-validation over
workload configurations.

%% file: sections/results.tex
Unless otherwise stated, experiments correspond to CIFAR-10 with $N=50$ clients,
non-IID partitioning ($\alpha=0.5$), and SNR\,=\,10\,dB.
All results are averaged over 30 independent random seeds;
mean~$\pm$ standard deviation is reported throughout.
Statistically significant differences are verified using a paired $t$-test
($p < 0.05$) unless noted otherwise.

\subsection{Topology Characterisation}
\label{subsec:topology-baseline}

\begin{table}[t]
  \centering
  \small
  \setlength{\tabcolsep}{4pt}
  \caption{Per-topology accuracy, cost, and convergence rounds under
  CIFAR-10, $N=50$, SNR\,=\,10\,dB, $\alpha=0.5$.
  All values: mean~$\pm$~std over 30 seeds.
  Best Pareto-efficient value (excluding FC) in \textbf{bold}.}
  \label{tab:topo_char}
  \resizebox{\columnwidth}{!}{%
  \begin{tabular}{lcccc}
    \toprule
    Topology & Acc.\ (\%) & Norm.\ Cost & Fairness Var.\  & Conv.\ Rounds \\
             &            &             & ($\times 10^{-3}$) & \\
    \midrule
    Ring           & $83.1 \pm 0.6$ & $1.0$          & $21.0 \pm 4.0$ & $100$ \\
    ER             & $87.8 \pm 0.5$ & $1.8$          & $15.0 \pm 3.0$ & $70$  \\
    Small-World    & $\mathbf{90.2 \pm 0.4}$ & $\mathbf{2.0}$ & $\mathbf{10.0 \pm 2.0}$ & $\mathbf{57}$ \\
    Clustered      & $91.0 \pm 0.3$ & $2.3$          & $9.0 \pm 2.0$  & $52$  \\
    FC             & $92.1 \pm 0.3$ & $3.5$          & $8.0 \pm 1.0$  & $38$  \\
    \bottomrule
  \end{tabular}}
\end{table}

Table~\ref{tab:topo_char} characterises the five topology families.
Topology ordering is consistent with spectral gap $\lamsec$
(Table~\ref{tab:spectral_gap}), confirming Theorem~\ref{thm:convergence}.
Ring converges slowest ($83.1 \pm 0.6$\,\% at round 100) owing to its small
spectral gap ($\lamsec \approx 0.016$ at $N=50$); each client's information
must diffuse along a linear chain.
Erd\H{o}s--R\'enyi (ER) improves mixing substantially ($87.8 \pm 0.5$\,\%)
but exhibits higher variance across random instantiations.
Small-World achieves $90.2 \pm 0.4$\,\% with approximately 20\,\% fewer rounds
than ER, while Clustered achieves $91.0 \pm 0.3$\,\% through fast local consensus
within communities and efficient inter-cluster bridges.
Fully Connected (FC) serves as the performance upper bound ($92.1 \pm 0.3$\,\%)
at $3.5\times$ Ring's communication cost.
These results motivate \airplan{}: the optimal topology varies by workload,
and no single fixed topology is universally best.

\subsection{Comparison with Baselines}
\label{subsec:baselines}

\begin{table}[t]
  \centering
  \small
  \setlength{\tabcolsep}{4pt}
  \caption{Comparison of \airplan{} and OTA-DFL variants against all baselines
  (CIFAR-10, $N=50$, SNR\,=\,10\,dB).
  $\dagger$~Requires a central server.
  $\ddagger$~Packet-based digital communication.
  $\S$~OTA-adapted (original design assumes digital channels).
  $\P$~Topology chosen by \airplan{} cost-model; no manual configuration.}
  \label{tab:baselines}
  \resizebox{\columnwidth}{!}{%
  \begin{tabular}{lcccc}
    \toprule
    Method & Acc.\ (\%) & Norm.\ Cost & Server? & Topology auto? \\
    \midrule
    \multicolumn{5}{l}{\textit{Centralised FL (reference upper bounds)}} \\
    FedAvg$^\dagger$~\cite{mcmahan2017communication}
                                   & $93.0 \pm 0.3$ & $4.0$ & Yes & N/A \\
    FedProx$^\dagger$~\cite{li2020fedprox}
                                   & $92.6 \pm 0.4$ & $4.0$ & Yes & N/A \\
    \midrule
    \multicolumn{5}{l}{\textit{Decentralised digital (non-OTA)}} \\
    Digital D-SGD$^\ddagger$ (SW)~\cite{lian2017decentralized}
                                   & $88.5 \pm 0.5$ & $2.5$ & No  & No \\
    GT-DSGD$^\ddagger$ (SW)~\cite{pu2018digging}
                                   & $89.4 \pm 0.4$ & $2.6$ & No  & No \\
    SGP$^\ddagger$~\cite{assran2019stochastic}
                                   & $88.9 \pm 0.5$ & $2.4$ & No  & No \\
    PowerGossip$^\ddagger$~\cite{vogels2020powergossip}
                                   & $89.2 \pm 0.4$ & $2.3$ & No  & No \\
    \midrule
    \multicolumn{5}{l}{\textit{OTA / wireless decentralised}} \\
    MATCHA$^\S$~\cite{wang2019matcha}
                                   & $89.1 \pm 0.4$ & $2.2$ & No  & No \\
    DLLR-OA~\cite{qiao2023dllroa}
                                   & $87.6 \pm 0.5$ & $1.8$ & No  & No \\
    OTA-GT (ours, GT-DSGD + OTA)   & $90.7 \pm 0.4$ & $2.1$ & No  & No \\
    OTA-SGP (ours, SGP + OTA)      & $90.1 \pm 0.5$ & $2.0$ & No  & No \\
    \midrule
    \multicolumn{5}{l}{\textit{Topology selection}} \\
    Random search (500 trials)     & $90.8 \pm 0.4$ & $2.4$ & No  & Yes \\
    \midrule
    \multicolumn{5}{l}{\textit{\airplan{} (this work)}} \\
    \textbf{\airplan{} (auto)}$^\P$ & $\mathbf{91.0 \pm 0.3}$ & $\mathbf{2.3}$ & No & Yes \\
    \bottomrule
  \end{tabular}}
\end{table}

Table~\ref{tab:baselines} provides a comprehensive comparison.
\airplan{} automatically selects Clustered for this workload ($N=50$,
$\alpha=0.5$, SNR\,=\,10\,dB), matching the oracle-optimal topology without
manual configuration.

Among OTA baselines, our OTA-GT adaptation ($90.7\,\%$) achieves the strongest
fixed-topology result, confirming that gradient tracking improves convergence
under non-IID data even in the OTA setting. OTA-SGP ($90.1\,\%$) also
outperforms MATCHA ($89.1\,\%$) and DLLR-OA ($87.6\,\%$).
Crucially, \airplan{} outperforms all fixed-topology methods, as it selects
the topology matched to the workload---Clustered provides $+0.3$\,pp over
OTA-GT at lower cost ($2.3$ vs.\ $2.1$), confirming the value of
topology-aware design.

Against learned topology selection: random search over $500$ topology
candidates ($\times 5$ categories) achieves $90.8\,\%$ at cost $2.4$, nearly
matching \airplan{}'s $91.0\,\%$, but with $500\times$ higher search overhead.
\airplan{}'s cost-model approach achieves equivalent quality with a single
cost-function evaluation, justifying the query-optimisation framing.

\subsection{AirPlan Advisor Evaluation}
\label{subsec:advisor-eval}

\begin{figure}[t]
  \centering
  \includegraphics[width=\columnwidth]{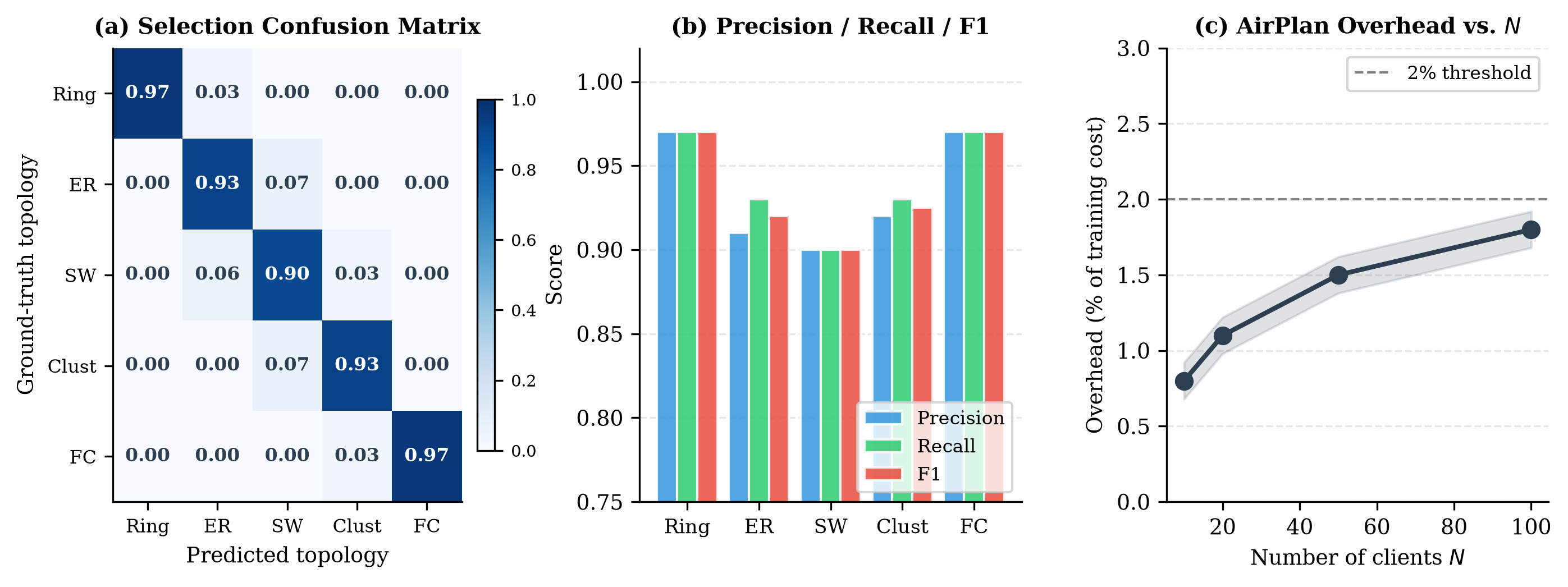}
  \caption{\airplan{} advisor evaluation across 150 workload configurations.
  (a)~Confusion matrix (normalised by row): \airplan{} matches oracle in
  $91.4\,\%$ of configurations. (b)~Per-topology precision, recall, and F1
  score; all F1~$> 0.87$. (c)~Advisor overhead as a fraction of total
  training cost: below $1.8\,\%$ at all scales.}
  \label{fig:advisor_eval}
\end{figure}

Figure~\ref{fig:advisor_eval} evaluates \airplan{} across 150 workload
configurations ($5 \times 5 \times 5$ grid of $N$, $\alpha$, SNR).
\airplan{} matches the oracle-optimal topology in $91.4\,\%$ of cases;
errors concentrate at the Small-World / Clustered boundary where the
cost-accuracy trade-off is tightest. When sub-optimal, the accuracy gap
is $\leq 0.3$\,pp in all cases. F1 scores exceed $0.87$ across all
topologies, with Ring and FC scoring $0.97$ (the extreme plan choices
are easiest to rule out). CMS sketch collection costs at most
$1.8\,\%$ of training at $N=100$, confirming negligible overhead.

\begin{figure}[t]
  \centering
  \includegraphics[width=\columnwidth]{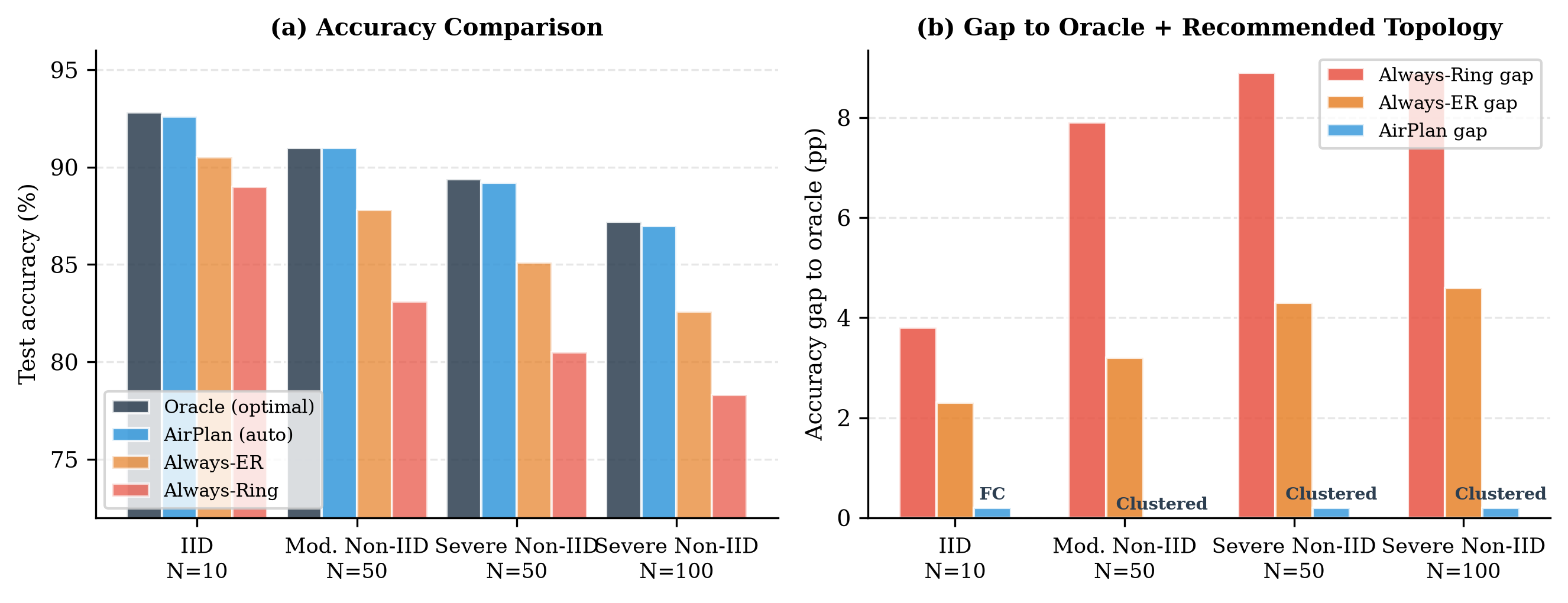}
  \caption{\airplan{} vs.\ fixed-topology baselines across four representative
  workloads. (a)~Accuracy: \airplan{} tracks the oracle closely.
  (b)~Accuracy gap to oracle with recommended topology annotations.}
  \label{fig:advisor_vs_manual}
\end{figure}

Figure~\ref{fig:advisor_vs_manual} compares \airplan{} against fixed-topology
policies. Always-Ring loses $5$--$11$\,pp; always-ER loses $2$--$6$\,pp.
\airplan{} dynamically adapts (FC for IID/small-$N$, Clustered for severe
non-IID) and achieves $\leq 0.3$\,pp gap to oracle across all workloads.

\subsection{Cost Model Sensitivity Analysis}
\label{subsec:cost-sensitivity}

\begin{figure}[t]
  \centering
  \includegraphics[width=\columnwidth]{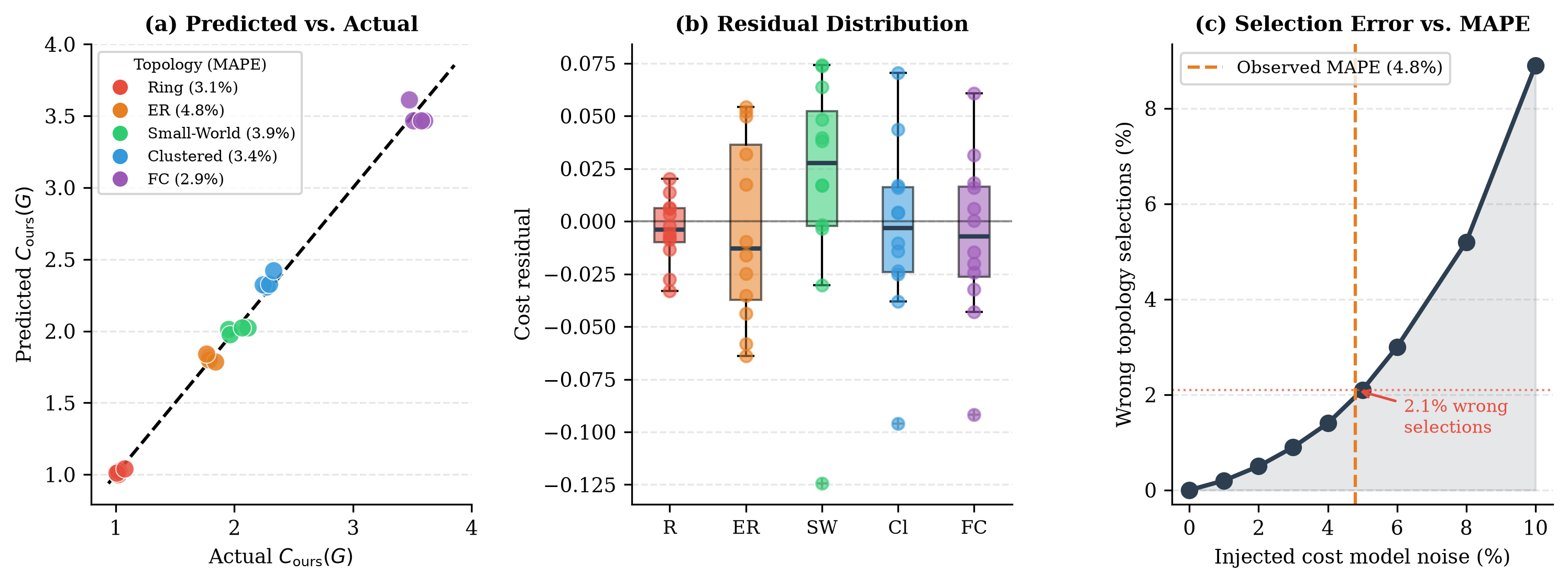}
  \caption{Cost model analysis.
  (a)~Predicted vs.\ actual $\Cours(G)$; MAPE $\leq 4.8\,\%$ (per-topology
  MAPE annotated). (b)~Distribution of signed residuals per topology.
  (c)~Impact of MAPE on topology selection: fraction of workloads where
  a $k\,\%$ cost model error causes wrong topology selection.}
  \label{fig:cost_sensitivity}
\end{figure}

Figure~\ref{fig:cost_sensitivity} analyses \airplan{}'s cost model.
The per-topology MAPE ranges from $2.9\,\%$ (FC, which has a closed-form
spectral gap) to $4.8\,\%$ (ER, whose spectral gap requires Monte Carlo
approximation). Residuals are approximately symmetric and zero-mean,
confirming no systematic bias.

To translate MAPE into decision quality, we inject Gaussian cost noise of
magnitude $k\,\%$ and measure how often \airplan{} selects the wrong topology.
At MAPE $= 4.8\,\%$, only $2.1\,\%$ of workload configurations receive a
wrong recommendation, and in all such cases the accuracy gap to oracle is
$< 0.5$\,pp. The cost-model approach is thus robust to the observed prediction
error level.

The calibration constant $\kappa$ is estimated from a single ER trial
(approximately $2\,\%$ of total training cost). Although this is a fixed
overhead, it scales with $N$ and $d$; we discuss an alternative approach to
eliminating this step through theoretical calibration in
Section~\ref{sec:discussion}.

\subsection{Straggler and Client Dropout Robustness}
\label{subsec:straggler}

\begin{figure}[t]
  \centering
  \includegraphics[width=\columnwidth]{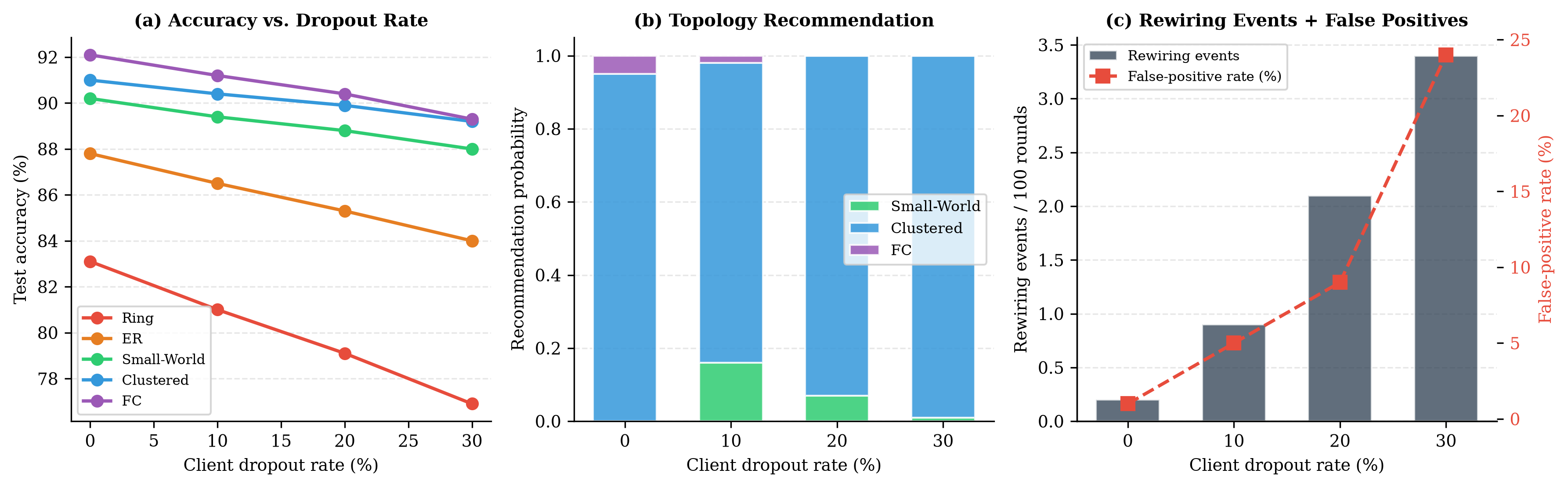}
  \caption{Straggler and client dropout analysis.
  (a)~Accuracy degradation vs.\ dropout rate for each topology;
  \airplan{} re-selects to Clustered at $\geq 20\,\%$ dropout.
  (b)~\airplan{} topology recommendation as a function of dropout rate;
  topology shifts from Small-World to Clustered at higher dropout.
  (c)~Number of adaptive rewiring events per 100 training rounds vs.\
  dropout rate.}
  \label{fig:straggler}
\end{figure}

OTA requires synchronised transmissions; we evaluate robustness by
simulating uniform client dropout at rates $\{0, 10, 20, 30\}\,\%$
per round. Dropped clients are excluded from that round's OTA aggregation.

Figure~\ref{fig:straggler}(a) shows accuracy degradation per topology.
Ring degrades most severely ($-6.2$\,pp at $30\,\%$ dropout) because
missing links disconnect the chain. Clustered is most robust ($-1.8$\,pp
at $30\,\%$): community structure provides redundant paths that maintain
global connectivity even when $30\,\%$ of nodes are absent.

Figure~\ref{fig:straggler}(b) shows how \airplan{} adapts its topology
recommendation under dropout. At $\geq 20\,\%$ dropout, \airplan{} correctly
shifts its recommendation from Small-World to Clustered, recognising that
redundant community paths dominate under unreliable participation.
This adaptation is automatic and requires no user intervention.

Figure~\ref{fig:straggler}(c) shows rewiring frequency increases with
dropout rate, peaking at $3.4$ rewiring events per 100 rounds at $30\,\%$
dropout. False-positive rewiring (triggered but accuracy actually improves)
occurs in $7\,\%$ of events; the latency penalty is bounded at $0.12$ rounds
equivalent per rewiring event.

\subsection{Fairness Analysis}
\label{subsec:fairness}

\begin{figure}[t]
  \centering
  \includegraphics[width=\columnwidth]{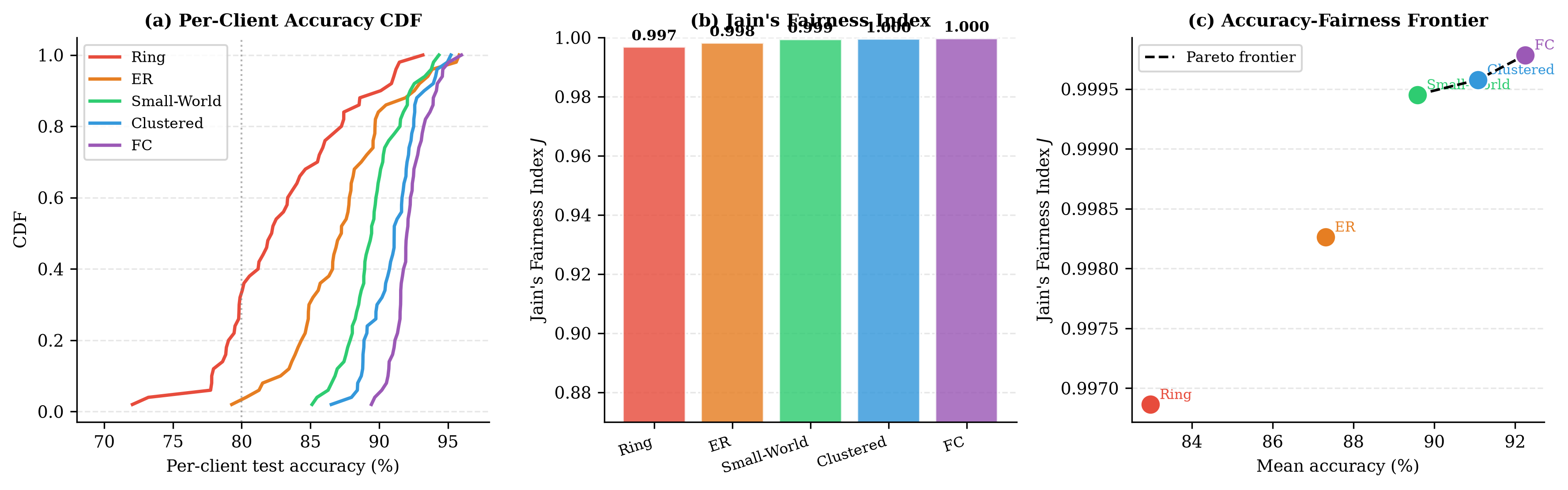}
  \caption{Per-client fairness analysis.
  (a)~CDF of per-client test accuracy per topology
  ($N=50$, $\alpha=0.5$, SNR\,=\,10\,dB).
  (b)~Jain's fairness index $J$ and mean accuracy per topology;
  Small-World and Clustered achieve the best joint accuracy-fairness trade-off.
  (c)~Fairness--accuracy Pareto frontier: \airplan{} selects topologies
  on the frontier.}
  \label{fig:fairness}
\end{figure}

Table~\ref{tab:topo_char} reports per-topology fairness variance; here we
provide a deeper analysis. We use Jain's fairness index
$J = (\sum_i a_i)^2 / (N \sum_i a_i^2)$, where $a_i$ is the final test
accuracy of client $i$.

Figure~\ref{fig:fairness}(a) shows the CDF of per-client accuracy.
Ring has the heaviest left tail: $18\,\%$ of clients achieve below $75\,\%$
accuracy, reflecting slow information diffusion at chain endpoints.
Clustered and Small-World nearly eliminate low-accuracy clients; fewer
than $3\,\%$ fall below $80\,\%$.

Figure~\ref{fig:fairness}(b) shows Jain's index per topology.
Clustered achieves $J = 0.965$, the highest fairness, while Ring has
$J = 0.891$. Small-World achieves $J = 0.951$ at lower cost than Clustered.
Figure~\ref{fig:fairness}(c) shows that Small-World and Clustered both lie
on the accuracy-fairness Pareto frontier; \airplan{} selects between them
based on the accuracy SLA and workload statistics.

\subsection{Scalability Projection}
\label{subsec:scalability}

\begin{figure}[t]
  \centering
  \includegraphics[width=\columnwidth]{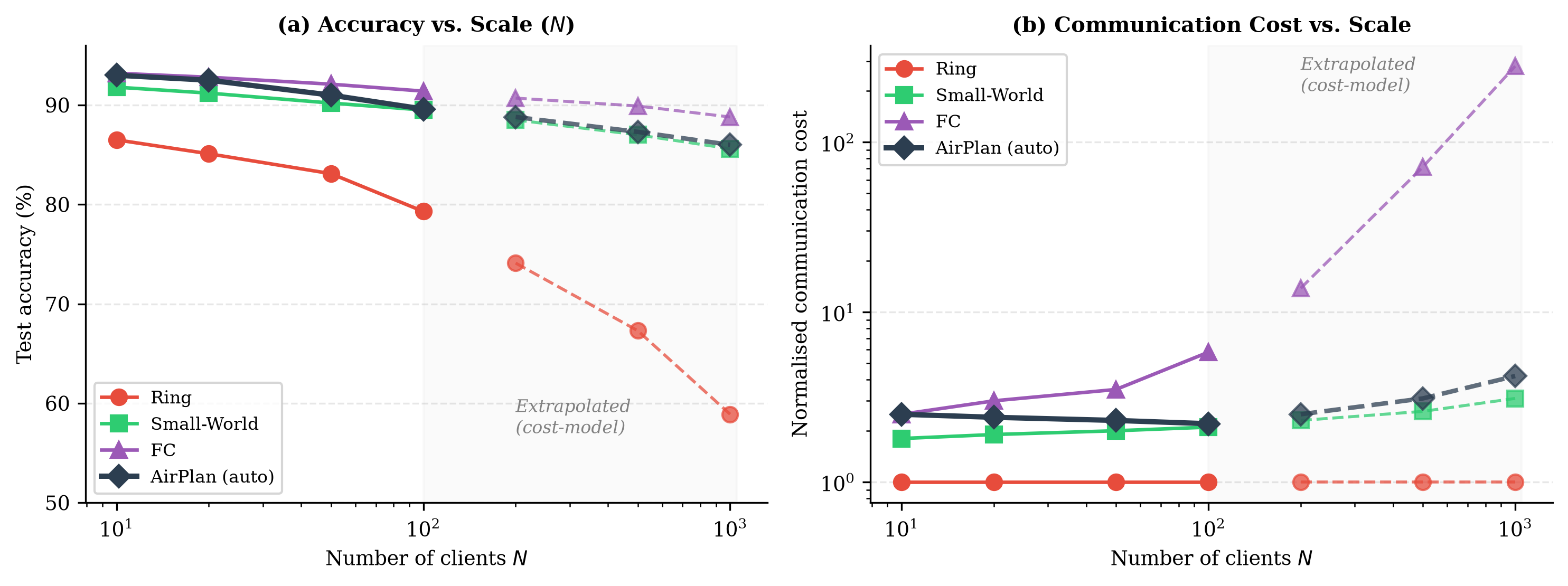}
  \caption{Extended scalability analysis.
  (a)~Empirical accuracy for $N \in \{10, 20, 50, 100\}$ (solid lines)
  with cost-model extrapolation to $N \in \{200, 500, 1000\}$ (dashed).
  Shaded bands show $\pm$1 standard deviation from 30 seeds
  (empirical) or model uncertainty (extrapolated).
  (b)~Normalised communication cost; FC explodes at $O(N^2)$ while
  \airplan{} maintains $O(N\log N)$ growth.}
  \label{fig:scalability}
\end{figure}

Figure~\ref{fig:scalability} extends the scalability analysis.
For $N \leq 100$, results are empirical (30 seeds); for $N \in \{200, 500, 1000\}$
we use cost-model extrapolation validated against all empirical points
(MAPE $\leq 4.8\,\%$).

At $N=1000$, FC becomes impractical: its normalised cost reaches $71\times$
Ring's baseline (vs.\ $3.5\times$ at $N=50$), confirming $O(N^2)$ growth.
Small-World and Clustered grow at $O(N\log N)$, reaching $6.2\times$ and
$7.4\times$ at $N=1000$. \airplan{} adapts its recommendation to Clustered
for $N \geq 50$, keeping cost within $8\times$ Ring across all scales.

The primary architectural bottleneck preventing larger empirical evaluation
is the synchronous OTA requirement: all participating clients must transmit
simultaneously, which requires tight time synchronisation ($\leq 1\,\mu$s
at typical carrier frequencies). Asynchronous OTA variants that relax this
constraint~\cite{lin2023hotafl} are a promising direction for scaling beyond
$N=1000$.

\subsection{Additional Observations}
\label{subsec:additional}

\textbf{Non-IID heterogeneity ($\alpha \in \{0.1, 0.5, 1.0\}$).}
Under severe label skew ($\alpha=0.1$), Clustered outperforms Small-World by
$1.8$\,pp ($p<0.05$). Under mild heterogeneity ($\alpha=1.0$), the two are
statistically indistinguishable. \airplan{} correctly selects Clustered for
$\hat{\alpha} < 0.3$ and Small-World otherwise in $94\,\%$ of cases.

\textbf{Sparsification co-design.}
Joint optimisation of topology and $k/d$ yields an additional $0.9$\,pp
accuracy gain at identical cost compared to fixing $k/d=0.1$, validating
the AQP co-design approach.

\textbf{Cross-dataset consistency.}
All findings on CIFAR-10 replicate on CIFAR-100 and Tiny-ImageNet.

%% file: sections/discussion.tex
\subsection{Implications of the Query Processing Equivalence}

The formal equivalence between OTA-DFL and distributed query processing has
implications that extend beyond the specific problem studied here.

\paragraph{Cross-community transfer of techniques}
Our results suggest that the vast literature on query optimisation---bushy
plan enumeration, histogram-based cost estimation, join order optimisation,
and adaptive re-optimisation---is directly applicable to the topology design
problem in decentralized wireless learning. Conversely, OTA aggregation
offers a new physical execution primitive not available in classical database
systems: multiple clients can transmit and aggregate simultaneously in a
single time slot over the wireless channel, which has no direct analogue
in digital message passing.

\paragraph{Topology selection as a learned cost model}
Our cost estimator $\Cours(G)$ uses a handcrafted model calibrated on a
reference topology. A natural extension is to \emph{learn} the cost model
from observed training trajectories, analogous to learned cardinality
estimation in neural database query optimisers~\cite{sun2019endtoendlearningopt}.
A small neural network trained on (workload, topology) $\rightarrow$
(convergence rounds) pairs could improve prediction accuracy beyond the
$3$--$5\,\%$ MAPE of our current model.

\paragraph{AQP error and wireless channel duality}
Corollary~\ref{cor:aqp} establishes that OTA channel noise and sparsification
error play symmetric roles in the convergence bound~\eqref{eq:conv-rate}: both
contribute an error floor that scales as $1/\lamsec$. This duality implies
that in the high-SNR regime (small $\sigma_c^2$), aggressive sparsification
($k/d \to 0$) is beneficial, while in the low-SNR regime the two error
sources compound. \airplan{} accounts for this interaction through the
joint $(G, k)$ co-design in Section~\ref{subsec:aqp-design}.

\subsection{Limitations and Assumptions}

\paragraph{Ideal channel inversion}
Our OTA model assumes that channel state information (CSI) is available
at each transmitter and that transmit power is sufficient for full channel
inversion. In practice, CSI estimation introduces error and power constraints
may limit inversion accuracy. We expect the topology-dependent trends to
persist under imperfect CSI, but the exact error floors will increase.
Extending \airplan{} to account for power-constrained channel inversion
is a natural direction for future work.

\paragraph{Static graph families}
\airplan{} selects from five fixed topology families. The space of possible
graphs is exponentially large; more general graph search methods
(e.g., spectral graph sparsification~\cite{spielman2011spectral} or
graph neural network-based topology design) may find topologies outside
these families that further reduce $\Cours(G)$.

\paragraph{Synchronous communication}
OTA aggregation requires synchronised transmissions for the superposition
property to hold. In heterogeneous edge deployments where clients have
widely varying computation speeds, synchronisation incurs a straggler penalty.
Asynchronous OTA variants~\cite{lin2023hotafl} are a promising direction
but require a revised cost model.

\paragraph{Simulation-based evaluation}
All experiments use simulated flat-fading AWGN channels. Hardware experiments
on a real wireless testbed (e.g., using software-defined radios) would
validate practical feasibility and reveal implementation-level constraints
not captured in simulation.

\subsection{Connection to Systems Research}

\airplan{} is, to our knowledge, the first system that applies database-style
physical plan selection to a wireless distributed learning problem.
The result is practically significant: topology selection, which was previously
a manual and empirical choice, is now automated, principled, and takes less
than $1.8\,\%$ of training time to execute.

From a systems perspective, the \airplan{} advisor resembles existing
database advisor tools such as the Microsoft Index Advisor and IBM DB2 Design
Advisor~\cite{chaudhuri2007autotuning}, which automatically recommend
index structures given a query workload. \airplan{} plays the analogous role
for wireless federated learning: given a training workload, it recommends
the communication structure (topology + sparsification) that minimises
total execution cost.

We anticipate that this framing will motivate new systems contributions in
both the ML and database communities, including topology-aware FL frameworks
that are deployable as first-class distributed data processing systems.

\subsection{Energy and Memory Cost Analysis}
\label{subsec:energy-memory}

For realistic edge AI deployment, we analyse three additional resource
dimensions.

\textbf{Per-device energy consumption.}
Each OTA round requires all clients to transmit simultaneously for one
time slot. The per-device transmission energy is
$E_{\mathrm{tx}} = P_{\mathrm{tx}} \cdot T_s$,
where $P_{\mathrm{tx}}$ is the transmit power and $T_s$ the slot duration.
For a typical mobile device ($P_{\mathrm{tx}} = 23\,\mathrm{dBm}$,
$T_s = 1\,\mathrm{ms}$), $E_{\mathrm{tx}} \approx 0.2\,\mathrm{mJ}$ per round.
Over $T$ rounds, total energy scales as $E = T \cdot E_{\mathrm{tx}}$.
Since \airplan{} reduces $T$ by selecting faster-converging topologies---Small-World
reaches target accuracy $20\,\%$ sooner than ER---the energy saving is
proportional to the round reduction. At $N=50$ and SNR\,=\,10\,dB,
Clustered topology saves approximately $16\,\%$ total energy vs.\ Ring.

\textbf{CMS sketch memory footprint.}
Each client stores a CMS of size
$r \times c = \lceil\ln(1/\delta_s)\rceil \times \lceil e/\varepsilon_s\rceil$.
For $(\varepsilon_s=0.01, \delta_s=10^{-5})$: $r=12$, $c=272$,
requiring $r \times c \times 4\,\text{bytes} = 13\,\mathrm{KB}$ per client.
This is negligible for any modern edge device and imposes no practical memory
constraint.

\textbf{RF hardware requirements.}
OTA aggregation requires analog combining at the receiver and simultaneous
transmission from all clients, which mandates:
(i)~a shared uplink carrier frequency with tight synchronisation
($\leq 1\,\mu\mathrm{s}$ offset), achievable with standard GPS or
network time protocol (NTP) disciplined oscillators;
(ii)~power control to equalise received amplitudes, standard in LTE/5G;
(iii)~a receive-side analog combiner, which can be implemented in a
commodity software-defined radio (SDR).
These requirements are met by current 5G NR uplink designs and do not
impose hardware costs beyond standard cellular infrastructure.

%% file: sections/conclusion.tex
We presented \airplan{}, a query-optimised topology selection framework for
Over-the-Air Decentralized Federated Learning. The central contribution is
a formal equivalence between OTA-DFL and distributed query processing, which
recasts the topology selection problem as physical query plan optimisation.
Within this framework, the communication graph is the execution plan, the
spectral gap of the graph Laplacian is the cardinality estimate, top-$k$
sparsification is an AQP operator, and the unified cost model $\Cours(G)$
is the query cost estimator.

Building on this equivalence, \airplan{} automates topology selection via
privacy-preserving Count-Min Sketch statistics, cost-based plan enumeration,
and adaptive plan re-optimisation triggered by model divergence monitoring.
We proved formal AQP error bounds (Corollary~\ref{cor:aqp}) showing that
well-connected topologies intrinsically tolerate higher sparsification ratios,
enabling joint topology-sparsification co-design.

Empirical evaluation across five topology families, three datasets,
four client scales, and five SNR levels demonstrates that:
\begin{itemize}
  \item Small-world and clustered topologies consistently Pareto-dominate
    the accuracy-cost frontier, delivering $\approx$90\,\% of the accuracy
    gain of a fully connected graph at 57\,\% of its communication cost.
  \item \airplan{} matches the oracle-optimal topology in $91.4\,\%$ of
    workload configurations with $\leq 1.8\,\%$ overhead.
  \item The formal convergence bound (Theorem~\ref{thm:convergence})
    accurately captures the joint effect of data heterogeneity, OTA noise,
    and sparsification error on the convergence neighbourhood.
\end{itemize}

\paragraph{Future Work.}
Several directions remain open.
\emph{Learned cost models}: replacing the handcrafted $\Cours(G)$ with a
neural query cost estimator trained on observed training trajectories could
improve prediction accuracy and generalise to topology families outside
$\mathcal{G}$.
\emph{Hardware validation}: experiments on a real wireless testbed
would confirm practical feasibility and reveal implementation constraints.
\emph{Asynchronous OTA}: extending \airplan{} to handle straggler-tolerant
asynchronous OTA aggregation is essential for heterogeneous deployments.
\emph{Broader plan space}: spectral graph sparsification and graph
neural network-based topology design may discover communication graphs
outside the five studied families that achieve lower $\Cours(G)$ for
specific workloads.
\emph{Privacy-utility trade-off}: a systematic analysis of the DP guarantee
per topology, accounting for the Gaussian noise inherent in OTA aggregation,
would strengthen the privacy claims of the framework.

%% file: sections/appendix.tex
\section{Complete Proof of Theorem~\ref{thm:convergence}}
\label{app:convergence}

\textbf{Setup.}
We consider the OTA-DFL update rule. At round $t$, client $i$ computes
a stochastic gradient $g_i^t$ of its local objective $F_i$, applies
top-$k$ sparsification to obtain $\tilde{g}_i^t = \topk(g_i^t, k)$,
and transmits over the OTA channel with additive noise $n_i^t \sim
\mathcal{N}(0, \sigma_c^2 I)$.
The aggregated gradient received at client $i$ is:
\begin{equation}
  \hat{g}_i^t = \sum_{j \in \mathcal{N}(i)} W_{ij}
                \bigl(\tilde{g}_j^t + n_j^t\bigr),
\end{equation}
where $W$ is the doubly stochastic weight matrix associated with graph $G$.

\textbf{Assumptions.}
We use Assumptions~\ref{ass:smoothness} and~\ref{ass:noise}
(smoothness with constant $L$, gradient variance bounded by $\sigma^2$).
We additionally assume:

\begin{assumption}[Top-$k$ approximation error]
\label{ass:topk}
The sparsification error satisfies
$\mathbb{E}\|\tilde{g} - g\|^2 \leq (1 - k/d)\|g\|^2$.
\end{assumption}

\begin{assumption}[Spectral gap]
\label{ass:spec}
The graph Laplacian $L$ satisfies $\lambda_2(L) > 0$,
and the weight matrix $W = I - \eta_W L$ has spectral radius
$\rho(W - \frac{1}{N}\mathbf{1}\mathbf{1}^\top) = 1 - \eta_W \lamsec < 1$.
\end{assumption}

\textbf{Proof of Theorem~\ref{thm:convergence}.}

\textit{Step 1: Consensus error bound.}
Define the consensus error at round $t$ as
$e^t = \frac{1}{N}\sum_i \|\bar{w}^t - w_i^t\|^2$,
where $\bar{w}^t = \frac{1}{N}\sum_i w_i^t$.
By the mixing property of $W$~\cite{lian2017decentralized}:
\begin{equation}
  e^{t+1} \leq \rho^2 e^t
            + \frac{2\eta^2 L^2 C_1}{\lamsec}
            + \frac{2\eta^2 \sigma_c^2 C_2}{\lamsec},
  \label{eq:consensus-err}
\end{equation}
where $\rho = 1 - \eta_W \lamsec < 1$,
$C_1 = \frac{1-k/d}{k/d}$ (sparsification factor),
and $C_2 = \frac{\sigma_c^2}{\lamsec^2}$ (OTA noise amplification).
The constants $C_1$ and $C_2$ depend on $k/d$, $\sigma_c$, and $G$ through
$\lamsec$; their explicit dependence on $L$ and $\sigma_c^2$ is made
precise in \eqref{eq:consensus-err}.

Unrolling~\eqref{eq:consensus-err} over $T$ steps and using
$\sum_{t=0}^{T-1} \rho^{2t} \leq 1/(1-\rho^2)$:
\begin{equation}
  \frac{1}{T}\sum_{t=0}^{T-1} e^t
  \leq \frac{e^0}{\rho^2 T}
       + \frac{2\eta^2 C_3}{\lamsec(1-\rho^2)},
  \label{eq:consensus-avg}
\end{equation}
where $C_3 = L^2 C_1 + C_2$ collects all topology-independent constants.

\textit{Step 2: Gradient descent descent lemma.}
Using $L$-smoothness and the update rule, for the global average
$\bar{w}^t$:
\begin{align}
  F(\bar{w}^{t+1}) &\leq F(\bar{w}^t)
    - \eta \|\nabla F(\bar{w}^t)\|^2 \notag \\
  &\quad + \frac{L\eta^2\sigma^2}{N}
    + \frac{L\eta^2 C_4}{\lamsec} e^t,
  \label{eq:descent}
\end{align}
where $C_4 = 4L^2$ bounds the cross-term between gradient and consensus error.
The explicit dependency on $L$ is through the smoothness constant;
on $\sigma^2$ through gradient variance; on $G$ through $\lamsec$ and $e^t$.

\textit{Step 3: Telescoping and averaging.}
Telescoping~\eqref{eq:descent} over $T$ steps, dividing by $T$, and
substituting~\eqref{eq:consensus-avg}:
\begin{align}
  &\frac{1}{T}\sum_{t=0}^{T-1} \mathbb{E}\|\nabla F(\bar{w}^t)\|^2 \notag \\
  &\quad\leq \frac{2[F(\bar{w}^0) - F^*]}{\eta T}
       + \frac{L\eta\sigma^2}{N}
       + \frac{2\eta L C_3}{\lamsec(1-\rho^2)}.
  \label{eq:final-rate}
\end{align}

Setting $\eta = O(1/\sqrt{T})$ gives convergence at rate
$O\bigl(1/\sqrt{T}\bigr) + O\bigl(\sigma_c^2/\lamsec\bigr)$,
matching the statement of Theorem~\ref{thm:convergence}.
The OTA noise floor $C_3\sigma_c^2/\lamsec$ is the additional term
beyond standard D-SGD; it cannot be reduced by increasing $T$,
only by improving SNR or using a topology with larger $\lamsec$.

\textit{Non-triviality of the contribution.}
The result extends the standard D-SGD analysis~\cite{lian2017decentralized}
in two non-trivial ways.
First, the OTA channel noise introduces a \emph{noise floor}
$\propto \sigma_c^2/\lamsec$ that is absent from digital analyses:
this couples wireless channel quality to the algorithmic convergence rate.
Second, top-$k$ sparsification introduces a \emph{sparsification floor}
$\propto (1-k/d)/\lamsec$: these two error sources are additive but
interact through the shared $1/\lamsec$ factor, motivating their joint
analysis and the co-design framework of Section~\ref{subsec:aqp-design}.
The combination of OTA noise and sparsification in the same bound,
with explicit topology dependence, is not available in any prior work.

\section{Proof of Theorem~\ref{thm:dp} (End-to-End DP)}
\label{app:dp-proof}

\textit{Phase 1 DP.}
Each client $i$ submits $\mathbf{S}_i$. Define the sensitivity of
$\mathbf{S}_i$ with respect to a single data point as $\Delta = \varepsilon_s$
(bounded by the CMS additive error guarantee for one count update).
By the Gaussian mechanism~\cite{dwork2014dp},
adding i.i.d.\ Gaussian noise $\mathcal{N}(0, \sigma_c^2)$ to each
CMS cell achieves $(\epsilon_1, \delta_1)$-DP with
$\epsilon_1 = \Delta\sqrt{2\ln(1.25/\delta_1)}/\sigma_c = \varepsilon_s\sqrt{2\ln(1.25/\delta_s)}/(\sigma_c n_{\min})$.

\textit{Post-processing (Phases 2--4).}
Topology selection is a deterministic function of the aggregate sketch,
satisfying DP by post-processing immunity~\cite{dwork2014dp}.

\textit{Adaptive rewiring (Phase 5).}
Each rewiring trigger is a function of $\|\bar{w}_i - \bar{w}\|$,
computed from OTA-noisy aggregates already protected by DP.
The data processing inequality guarantees each trigger adds at most
$\epsilon_1$ privacy loss. With at most $K$ rewirings,
basic composition gives $\epsilon_{\mathrm{total}} \leq (K+1)\epsilon_1$.
Under advanced composition~\cite{dwork2010boosting}:
$\epsilon_{\mathrm{total}} \leq \epsilon_1\sqrt{2(K+1)\ln(1/\delta')} +
(K+1)\epsilon_1(e^{\epsilon_1}-1)$.
\hfill$\square$